\documentclass{tufte-handout}
\pdfoutput=1 

\usepackage{pifont}
\usepackage{mdframed} 
\newmdenv[
  leftmargin = 0pt,
  innerleftmargin = 0em,
  innertopmargin = -10pt,
  innerbottommargin = 0pt,
  innerrightmargin = 0pt,
  rightmargin = 0pt,
  linewidth = 1.5pt,
  topline = false,
  rightline = false,
  bottomline = false
  ]{leftbar}
  
\usepackage[fleqn]{amsmath}
\usepackage{latexsym,amssymb,amstext,mathtools}
\usepackage[amsmath,thmmarks]{ntheorem}
\usepackage{microtype}  
\usepackage{epic,autograph}
\usepackage{booktabs,xspace}
\usepackage{natbib,caption,subcaption}
\usepackage{enumitem}
\usepackage[dvipsnames]{xcolor}  

\theoremstyle{plain}
\newtheorem{thm}{Theorem}
\newtheorem{theorem}[thm]{Theorem}
\newtheorem{corollary}[thm]{Corollary}
\newtheorem{lemma}[thm]{Lemma}

\theorembodyfont{\rmfamily}
\newtheorem{definition}[thm]{Definition}

\newtheorem{scholium}[thm]{Scholium}
\newtheorem{example}[thm]{Example}

\theoremheaderfont{\it}\theorembodyfont{\upshape} \theoremstyle{nonumberplain}
\theoremseparator{} \theoremsymbol{\rule{1ex}{1ex}} 
\newtheorem{proof}{Proof}
\newtheorem{proofsketch}{Proof sketch}

\newcommand{\CASE}{\medskip \textsc{Case:}\enspace}
\newcommand{\SCASE}{\medskip \textsc{Special case:}\enspace}
\newcommand{\CASES}{\medskip \textsc{Cases:}\enspace}

\newcommand{\safe}[1]{{#1^{\mathsf{S}}}}
\newcommand{\norm}[1]{#1^{\mathsf{N}}}

\newcommand{\key}[1]{\ensuremath{\mathop{\mathsf{#1}}\nolimits}\xspace}

\newcommand{\Case}{\key{case}}
\newcommand{\Of}{\key{of}}
\newcommand{\In}{\key{in}}
\newcommand{\Let}{\key{let}}

\newcommand{\fold}{\key{fold}}

\newcommand{\type}[1]{\ensuremath{{\bf #1}}\xspace}
\newcommand{\Unit}{\type{unit}}
\newcommand{\Nat}{\type{nat}}

\newcommand{\Tree}{\type{tree}}

\newcommand{\String}{\type{string}}

\newcommand{\size}{\ensuremath{\mathrm{size}}\xspace}

\newcommand{\toSafe}{\key{toSafe}}

\newcommand{\toNorm}{\key{toNorm}}

\newcommand{\constr}[1]{\ensuremath{\mathop{\mathsf{#1}}\nolimits}\xspace}

\newcommand{\Zero}{\constr{Zero}}
\newcommand{\Succ}{\constr{Succ}}
\newcommand{\Leaf}{\constr{Leaf}}
\newcommand{\Fork}{\constr{Fork}}
\newcommand{\Empty}{\constr{Empty}}

\newcommand{\Cons}{\constr{Cons}}

\newcommand{\Branch}{\constr{Branch}}

\newcommand{\Cn}[1]{\mathop{\mathsf{c}_{#1}}}
\newcommand{\Cs}[1]{\mathop{\mathsf{c}_{\safe{#1}}}}
\newcommand{\Dn}[1]{\mathop{\mathsf{d}_{#1}}}
\newcommand{\Ds}[1]{\mathop{\mathsf{d}_{\safe{#1}}}}
\newcommand{\uCn}[1]{\mathop{\underline{\mathsf{c}}_{#1}}}

\newcommand{\id}{\mathrm{id}}
\newcommand{\Id}{\mathrm{Id}}

\newcommand{\lfp}[1]{\mu #1\,\sqdot\,}

\newcommand{\SetCat}{\ensuremath{\mathcal{S}\mkern-2.5mu\mathit{et}}\xspace}

\newcommand{\ustrut}{\rule[-.18\baselineskip]{0pt}{0pt}}

\newcommand{\inj}{\mathop{\iota}\nolimits}

\newcommand{\proj}{\mathop{\pi}\nolimits}
\newcommand{\upair}[2]%
  {\underline{(}#1\underline{\ustrut,}\, #2\underline{)}}
\newcommand{\uinj}{\mathop{\underline{\iota}}\nolimits}
\newcommand{\uep}{\underline{()}}

\newcommand{\SB}{\mathcal{B}}

\newcommand{\SE}{\mathcal{E}}

\newcommand{\RSi}{\ensuremath{\mathit{RS_1}}\xspace}

\newcommand{\RSmi}{\ensuremath{\mathit{RS^-_1}}\xspace}

\newcommand{\iffdef}{\mathrel{\iff_{\mkern-12mu\text{def}}}}

\newcommand{\newfontobj}[2]{
  \newcommand{#1}[1]{
    \expandafter\def\csname##1\endcsname{{#2{##1}}}}}

\newfontobj{\class}{\textrm}
\newfontobj{\fnc}{\mathit}
\newfontobj{\category}{\mathbf}
\newfontobj{\Constructor}{\mathsf}

\newcommand{\Strut}[1]{\rule{0pt}{#1}}

\renewcommand{\gets}{\ensuremath{\mathrel{\colon=}}\xspace}

\newcommand{\dom}{\mathop{\mathrm{dom}}}

\newcommand{\set}[1]{\{\,#1\,\}}
\newcommand{\Set}[1]{\left\{\,#1\,\right\}}

\newcommand{\entails}{\vdash}
\renewcommand{\phi}{\varphi}

\newcommand{\suchthat}{\mathrel{\,\stackrel{\rule{0.03em}{0.5ex}}%
{\rule[-.1ex]{0.03em}{0.5ex}}\,}}

\newcommand{\sqdot}{\mathord{\rule{0.4mm}{0.4mm}}}
\newcommand{\lam}[1]{\lambda #1{\mkern1mu.\mkern1mu}}

\newcommand{\defeq}{\mathrel{\stackrel{\text{def}}{=}}}
\newcommand{\fv}{\ensuremath{\mathbf{fv}}}

\renewcommand{\colon}{\mathpunct{:}}
\newcommand{\of}{\mathbin{:}}
\newcommand{\ofx}{\mathbin{:\,}}
\newcommand{\is}{\ensuremath{\mathrel{\colon\colon=}}}
\newcommand{\synsep}{\mathrel{\,|\,}}
\newcommand{\yields}{\mathbin{\downarrow}}
\newcommand{\Quad}[1]{\hspace*{#1em}}

\renewcommand{\today}{\number\day \space \ifcase\month\or
	January\or February\or March\or April\or May\or June\or
	July\or August\or September\or October\or November\or
	December\fi \space \number\year}

\newcommand{\Comment}[1]{\hbox{\textsl{/\!/ #1}}}

\newcommand{\irule}[2]{\ensuremath{{\frac{\strut\textstyle #1}%
  {\strut\textstyle #2}}}} 
\newcommand{\rulelabel}[1]{\hbox{\itshape #1:}\Quad{0.5}}
\newcommand{\sidecond}[1]{\Quad{0.5} \left(#1\right)}
  
\fnc{depth}   \fnc{length} \fnc{pos}
\newcommand{\codepth}{\textit{co\kern-1ptDepth}}

\newcommand{\xs}{\mathit{xs}}

\newcommand{\step}{\ensuremath{\mathit{step}}\xspace}

\let\oldmarginpar\marginpar
\renewcommand\marginpar[1]{\-\oldmarginpar[\raggedleft\footnotesize #1]%
{\raggedright\footnotesize #1}}

\newcommand{\nat}{\mathbb{N}}
\newcommand{\Ltree}{\type{ltree}}
\newcommand{\Lst}{\type{list}}
\newcommand{\List}{\type{list}}

\fnc{step}\fnc{seed}\fnc{drop}\fnc{tail}\fnc{head}\fnc{nth}
\fnc{final}\fnc{steps}
\fnc{squares}\fnc{zeros}\fnc{mapSqr}
\fnc{pad}
\fnc{BC}
\fnc{trail}
\fnc{map}
\fnc{plantation}\fnc{graft}\fnc{harvest}\fnc{join}

\fnc{rest}

\newcommand{\dotcup}{\ensuremath{\mathbin{\mathaccent\cdot\cup}}}

\def\dashedrule#1#2#3{{%
  \dimen1=#2 \divide\dimen1 by 2
  
  \def\@ruledash{%
    \rule{\dimen1}{0pt}%
    \rule[0.5ex]{#1}{0.4pt}%
     \rule{\dimen1}{0pt}}%
     
    \count1=0
    \loop%
    \ifnum\count1<#3%
       \advance\count1 by 1%
       \@ruledash%
    \repeat}}

\let\boundedby\sqsubseteq
\def\vboundedby{\boundedby^{\mathrm{val}}}
\newcommand{\bounded}[3][\relax]{\ifx#1\relax{#2}\boundedby{#3}\else{#2}\boundedby_{#1}{#3}\fi}
\newcommand{\vbounded}[3][\relax]{\ifx#1\relax{#2}\vboundedby{#3}\else{#2}\vboundedby_{#1}{#3}\fi}

\def\oftype{\mathrel:}

\let\evalto\downarrow
\def\evalnohyp#1#2{{#1}\evalto{#2}}
\def\evalhyp#1#2#3{{#1}\entails\evalnohyp{#2}{#3}}
\newcommand{\eval}[3][\relax]{\ifx#1\relax\evalnohyp{#2}{#3}\else\evalhyp{#1}{#2}{#3}\fi}

\def\typenohyp#1#2{{#1}\oftype{#2}}
\def\typehyp#1#2#3{{#1}\entails\typenohyp{#2}{#3}}
\newcommand{\typing}[3][\relax]{\ifx#1\relax\typenohyp{#2}{#3}\else\typehyp{#1}{#2}{#3}\fi}

\newcommand{\compl}[1]{\overline{#1}}

\newcommand{\cost}[1]{\lceil #1 \rceil}

\newcommand{\seq}[2][\relax]{
  \ifx#1\relax
    \langle#2\rangle
  \else
    \ifx#1\left
      \left\langle#2\right\rangle
    \else
      \csname #1l\endcsname\langle#2\csname #1r\endcsname\rangle
    \fi
  \fi
}

\renewcommand{\set}[2][\relax]{
  \ifx#1\relax
    \{#2\}
  \else
    \ifx#1\left
      \left\{#2\right\}
    \else
      \csname #1l\endcsname\{#2\csname #1r\endcsname\}
    \fi
  \fi
}

\fnc{bump}

\makeatletter
\newcommand*\rel@kern[1]{\kern#1\dimexpr\macc@kerna}
\newcommand*\widebar[1]{%
  \begingroup
  \def\mathaccent##1##2{%
    \rel@kern{0.8}%
    \overline{\rel@kern{-0.8}\macc@nucleus\rel@kern{0.2}}%
    \rel@kern{-0.2}%
  }%
  \macc@depth\@ne
  \let\math@bgroup\@empty \let\math@egroup\macc@set@skewchar
  \mathsurround\z@ \frozen@everymath{\mathgroup\macc@group\relax}%
  \macc@set@skewchar\relax
  \let\mathaccentV\macc@nested@a
  \macc@nested@a\relax111{#1}%
  \endgroup
}
\makeatother

\newcommand{\vecx}{\vec{x}}
  
\newcommand{\uclos}[2]%
  {\underline{(\!(}#1\underline{\ustrut,}\, #2\underline{)\!)}}
  
\renewcommand{\Id}{\mathsf{Id}}
\newcommand{\App}{\textsf{App}}
\newcommand{\Pair}{\textsf{Pair}}
\newcommand{\Proj}{\textsf{Proj}}
\newcommand{\Inj}{\textsf{Inj}}

\newcommand{\VAL}{\textsf{V}}
\newcommand{\Var}{\textsf{Variable}}
\newcommand{\Destr}{\textsf{Destr}}  
\newcommand{\datatype}{\mathsf{datatype}\;}
\newcommand{\nn}[1]{\overline{\sharp_{#1}}}
\newcommand{\emptydag}{\mathbf{K}_0}

\newcommand{\Smi}{\ensuremath{S_1^-}\xspace}

\newcommand{\Grd}{\type{grd}}

\newcommand{\Alist}{\type{alist}}
\newcommand{\State}{\type{state}}

\newcommand{\numeral}[1]{\overline{#1}}
\newcommand{\RSmii}{\ensuremath{\mathit{RS}_{1.1}^-}\xspace}
\renewcommand{\size}{\mathord{\underline{\mathit{size}}}}
\newcommand{\ts}{\mathord{\underline{\mathit{ts}}}}
\newcommand{\cs}{\mathord{\underline{\mathit{cs}}}}
\newcommand{\note}[1]{\makebox[0pt][l]{\textsc{#1 rules}}}
\fnc{extract}
\fnc{equal}
\fnc{search}
\fnc{step}
\fnc{height}
\fnc{grow}
\fnc{plus}
\fnc{treeSize}
\fnc{tally}
\fnc{lookup}
\fnc{incr}
\fnc{bnd}
\fnc{serialize}
\fnc{deserialize}
\fnc{cost}
\newcommand{\CS}[1]{\mathop{\mathsf{cs}_{#1}}}
\newcommand{\cquote}[1]{\mathopen\ulcorner\! #1 \!\mathclose\urcorner}
\newcommand{\Const}{\mathsf{Constr}}
\newcommand{\VERT}{\mathsf{Vert}}
\newcommand{\Term}{\type{term}}
\newcommand{\Vtg}{\type{vtg}}
\newcommand{\Context}{\type{context}}
\newcommand{\Record}{\type{record}}
\newcommand{\Kont}{\type{kont}}
\newcommand{\Vertex}{\mathbf{vertex}}
\newcommand{\muP}{{\mu P}}
\newcommand{\Max}{\mathord{\mathit{max}}}

\newcommand{\BX}{\mathbf{X}}
\newcommand{\kids}[2]{
  \setprofcurve{-2} 
  \drawcurvedtrans[r](#1,#2){\scriptsize }
  \setprofcurve{2} 
  \drawcurvedtrans(#1,#2){\scriptsize }}

\newcommand{\residual}[1]{\Vert #1\Vert}
\newcommand{\TO}{&\quad\to\quad&}

\newcommand{\rulenum}[1]{\text{(R#1)}}

\title[General Ramified  Recurrence and Polynomial-time Completeness]{General Ramified  Recurrence and Polynomial-time Completeness
 (Preliminary Draft)}
\author{Norman Danner and James S.~Royer}
\begin{document}

\maketitle \marginnote{Norman Danner, 
         Dept.~of Mathematics and Computer Science, 
         Wesleyan University, Middletown, CT 06459 USA.  
         Email: \textsf{ndanner@wesleyan.edu}
         \medskip
         
         \noindent
         James S.~Royer, Dept.~of Elec.~Engrg.~and Computer Science, 
         Syracuse University, Syracuse, NY 13244 USA.
         Email: \textsf{jsroyer@syr.edu} 
         }

\vspace*{-2ex}
\begin{abstract}
  We exhibit a sound and complete implicit-complexity formalism for 
  functions 
  feasibly computable by structural recursions over inductively 
  defined data structures.
  \emph{A feasibly computable structural recursion} 
  here means that the structural-recursive definition has a run time 
  that is polynomial in the sizes of the representation of the  data 
  inputs
  and where these representations may make use of data sharing.  
  \emph{Inductively defined data structures} 
  here includes lists and trees.  
  \emph{Soundness} 
  here means that the  programs in the
  implicit-complexity formalism have feasible run times.
  \emph{Completeness}
  here means that each function computed by a feasible structural
  recursion has a program in the implicit-complexity formalism.
  This paper is a follow up on the work of
  \citep{gen:rem:2010,Avanzini2018OnSM} who focused on the soundness
  of such formalisms but did not consider the question of
  completeness.  
  \marginnote[-1cm]{\emph{Acknowledgements:}
Norman Danner's work was supported by NSF grant number 1618203
and James Royer's was supported by NSF grant number 1319769.
The second author wishes to thank his cat Penny for not stomping on
his keyboard too hard or too often.}
  
\end{abstract}


\section{Introduction}\label{S:intro} 

\citet{gen:rem:2010} and \citet{Avanzini2018OnSM} studied what
``feasibly computable'' should mean with respect to structural
(primitive) recursions over inductively defined data and how such
feasibly computable functions can be captured within an
implicit-complexity formalism.  Their focus was on showing the
soundness of such a formalism (i.e., that the programs of the
formalism have feasible runtimes). The present paper's focus is
on identifying such a formalism that is both sound and
\emph{complete}, i.e., that each feasibly-computable structurally
recursive function has a program in the formalism.

To illustrate the issues involved, let us consider some informal
examples.  First we introduce two data types using an ML-like
syntax:
\begin{gather*}
  \datatype \Nat \,= \Zero 
  \;|\; \Succ \Of\Nat\\
  \datatype \Tree = \Leaf 
  \;|\; \Branch \Of \Tree\times\Tree  
\end{gather*}
where $\Nat$ is a data-type for the natural numbers with
constructors $\Zero \ofx \Nat$ and $\Succ\ofx \Nat\to\Nat$ and
$\Tree$ is a data-type for binary trees with constructors
$\Leaf\ofx\Tree$ and $\Branch \ofx\Tree\times\Tree \to\Tree$.
Now, the variety of functional programming we are considering is
call-by-value with heap allocated data structures with structure
sharing.  For example, consider the following primitive recursive
definition.
\begin{align*}
  &\grow\of\Nat\to\Tree
  \\
  &\grow \,\Zero  = \Leaf
  \\
  &\grow \,(\Succ n)  = \big(\;\Let \, t = \grow\,n \,\In\,\Branch(t,t)\;\big)
\end{align*}
The result of evaluating\sidenote[][-2cm]{%
   \emph{Conventions:} Let $\numeral{m}= \Succ^{(m)}(\Zero)$ for each
  $m\in\nat$; i.e., $\numeral{0}=\Zero$, $\numeral1$ = $\Succ(\Zero)$,   
  $\numeral{2}=\Succ(\Succ(\Zero))$, etc. 
   We assume:\\ 
   $\plus(\numeral{m_1},\numeral{m_2}) = \numeral{m_1+m_2}$. \\
   $\Max(\numeral{m_1},\numeral{m_2})  = \numeral{\max(m_1,m_2)}$.
}  $\grow(\overline{3})$ is a directed acyclic
graph (dag) along the lines of the one shown in
Figure~\ref{fig:grow3}.
\begin{marginfigure}[1cm]
\setlength{\unitlength}{2.65pt}
{\ }\hfill
\begin{picture}(5,50)(-2.5,5)
\thicklines \setstatediam{10}\setrepeatedstatediam{9}
\letstate C1=(0,58) \drawstate(C1){\scriptsize$\Branch$}
\letstate C2=(0,42) \drawstate(C2){\scriptsize$\Branch$}
\kids{C1}{C2}
\letstate C3=(0,26) \drawstate(C3){\scriptsize$\Branch$}
\kids{C2}{C3}
\letstate C4=(0,10) \drawstate(C4){\scriptsize$\Leaf$}
\kids{C3}{C4}
\end{picture}\hfill{\ }
%
\caption{$\grow(\overline{3})$'s dag.}\label{fig:grow3}
\end{marginfigure} 
This four-vertex dag represents a complete binary tree of height three 
by means of structure sharing.  This compressed representation does not 
cause any troubles for our  programs provided, as in ML and Haskell,
our programs have no means to distinguish this representation 
from a 15-vertex sharing-free representation
of the same tree. 

These differing representations pose  some puzzles 
regarding the run-time complexity of programs over them.  
Consider the following two structural-recursive 
definitions.
\begin{align*}
  & \treeSize  \of\Tree\to\Nat
  \\
  &\treeSize \,\Leaf  = \numeral1
  \\
  &\treeSize \,(\Branch(tl,tr)) 
     = \Succ(\plus(\treeSize\, tl,\treeSize\, tr))
\\[1ex]
  & \height \of\Tree\to\Nat
  \\
  &\height \,\Leaf = \numeral0
  \\
  &\height \,(\Branch(tl,tr)) = \Succ(\Max(\height\,tl,\height\,tr))
\end{align*}
For any given $t\of\Tree$, $\treeSize(t)$ = the number of vertices in
the binary tree represented by $t$ and $\height(t)$ = the height of 
$t$'s binary tree.   Thus, for any $m\in\nat$:
\begin{align*}
 &\treeSize(\grow(\numeral{m})) = \numeral{2^{m+1}-1}.
 \\
 &\height(\grow(\numeral{m})) = \numeral{m}.
\end{align*}

Let us consider $\lam{n}\treeSize(\grow(n))$.  As there is an
exponential blow-up,  $\lam{n}\treeSize(\grow(n))$ fails to
be feasible.  We want our notion of feasibility to be closed under
composition, so it follows that at least one of $\treeSize$ and
$\grow$ should also count as infeasible.  If we classify $\grow$ as
infeasible, then the root cause of the infeasibility would seem to
be data sharing.  It \emph{is} possible to count $\treeSize$ as
feasible \emph{provided}, as in \citep{Burrell:2009}, one forbids
data sharing.  But abandoning data-sharing would be contrary to
long-standing, well-founded practices in functional programming with
which we would like to be consistent.  Thus, we seem to be forced to
count $\treeSize$ as infeasible.

Now let us consider $\lam{n}\height(\grow(n))$.  The function
computed by $\lam{n}\height(\grow(n))$, namely $n\mapsto n$, is
feasible if anything is.  However under standard evaluation
strategies, the computations specified by $\lam{n}\height(\grow(n))$
are infeasible.  The problem is that these strategies are oblivious
to data sharing and so, when fed the result of $\grow(\numeral{m})$,
the resulting computation goes through $(2^{m+1}-1)$-many calls of
$\height$.  A key insight in \citep{gen:rem:2010,Avanzini2018OnSM}
is that it is better to treat these structural recursions as dynamic
programming problems to be evaluated bottom-up rather than in the
standard top-down fashion.  In a dynamic programming evaluation of
$\height(t)$, the underlying dag\sidenote[][-2ex]{%
  ``Every dynamic program has an underlying dag structure: think of
  each node as representing a subproblem, and each edge as a
  precedence constraint on the order in which the subproblems can be
  tackled.  Having nodes $u_1,\dots,u_k$ point to $v$ means
  `subproblem $v$ can only be solved once the answers to
  $u_1,\dots,u_k$ are known.''' {\citep[Page~163]{DPV}} }
is just $t$'s dag.  The evaluation starts with the $\Leaf$ vertices
and works its way up the dag so that when considering a vertex
$\Branch(tl,tr)$ and computing
$\Succ(\Max(\height\,tl,\height\,tr))$ the values of $\height\,tl$
and $\height\,tr$ have already been computed and saved (memoized) so
that the recursive calls $\height\,tl$ and $\height\,tr$ turn into
look-ups.  Let us call the determination of $\height$ for a
$\Branch$ or a $\Leaf$ vertex of $t$ a \emph{recursive step} and
also define $\size(t)$ = the number of vertices in $t$'s
dag.\sidenote{%
  \emph{Convention:} Functions and operations on dags that, 
  unlike our definition of $\height$,   
  are cognizant of the concrete dag structure
  are \underline{underlined}.  Also, the formal definition of 
  $\size$ (Definition~\ref{d:size}) has some key differences
  from this informal version.  
  }
Thus, in computing $\height(t)$, the number of recursion steps is
$\size(t)$.
Moreover, all the $\Nat$s produced during the course of this
computation are of size at most $\size(t)$ and the computations
involving them (i.e., $\Succ(\Max(nl,nr))$) are clearly
polynomial-time in the sizes of $nl$ and $nr$.  Thus, under the
dynamic programming evaluation scheme $\lam{t}\height(t)$, and thus
$\lam{n}\height(\grow(n))$, appear to be feasibly computable (i.e., 
computable within time polynomial in $\size(t)$).  Note that
$\lam{n}\treeSize(\grow(n))$ remains infeasible under dynamic
programming evaluation as it still involves an exponential blow-up.

The focus of \citet{gen:rem:2010} and \citet{Avanzini2018OnSM} was
to use their insight on dynamic programming evaluation of structural
recursions to help resolve a long-standing problem in implicit
complexity theory.  \emph{Implicit complexity theory} seeks to
characterize computational complexity classes via restricted
programming or logic formalisms.  Standard characterizations of
complexity classes usually involve a low level machine model and
explicit resource bounds based on the size of the machine's input.
In contrast, an implicit complexity characterization captures a
complexity class without reference to either a machine model or
explicit resource bounds.  Implicit complexity techniques worked
well to handle feasible structural recursions over \emph{sequential}
data types (e.g., strings, lists), but seemingly broke down on
structural recursions on \emph{branching} data types (e.g., $\Tree$
above).  \citet{gen:rem:2010} and \citet{Avanzini2018OnSM} showed
what was needed was a shift to a cleverer operational semantics.
They demonstrated this by generalizing Leivant's
\citeyearpar{Leivant:FM2} early work on tiered/ramified  
formalisms to produce
tiered functional algebras over structured data that could 
handle branching structural recursions under what we have
been calling the dynamic programming evaluation.  Moreover, 
they formalized the dynamic programming evaluation scheme in a rewrite
system which they proved polynomial-time sound.

A question not addressed in \citep{gen:rem:2010,Avanzini2018OnSM} was
that of completeness, that is, 
\emph{does their formalism have a program
for each function computable by a feasible structural recursion? } 

We conjecture, no.  In particular, we suspect that the $\height$
function defined above is not computable in their formalism.
However, their formalism can easily compute
$\mathit{heightMod}\of \Tree\times \Nat \to\Nat$ such that
\begin{align*}
   &\mathit{heightMod}(t,n) = (\text{the height of $\Tree$ $t$}) \bmod n.
\end{align*}
Thus the problem comes down to extracting from $t$ any sort of 
upper bound on its height; that seems hard if not impossible
within the constraints of their formalism.

\subsection*{Our results}

We present a sound and complete implicit-complexity formalism for 
feasibly computable structural recursions over inductively 
defined data structures.  We do this as follows.
\begin{itemize}

  \item
First, \S\ref{S:Smi} introduces the 
formalism $\Smi$ for structural (primitive) recursions on inductively 
defined data together with introducing two operational semantics for $\Smi$,
\begin{description}
   \item[TD:] for
the standard top-down evaluations of structural recursions
and \item[DP:] for the dynamic-programing  evaluations of structural 
recursions.  
\end{description}
We also introduce a notion of \emph{cost} for both the 
TD and DP 
operational semantics. The runtime complexity of $\Smi$ programs
as measured by our notions of cost will be obviously polynomially 
related to the runtimes of said programs under a straightforward
RAM-based interpreter for $\Smi$ under the appropriate operational semantics.   Thus polynomial-cost under the TD and the DP operation semantics gives us two robust notions of feasible structural
recursion over inductively defined data.  

  \item
Then \S\ref{S:RSmi} introduces $\RSmi$, a normal/safe 
\citep{BellantoniCook} ramified version of $\Smi$.  Despite surface 
differences, $\RSmi$ is roughly comparable to the formalism of 
\citep{Avanzini2018OnSM}; see \S\ref{S:compare}.  

  \item 
In \S\ref{S:seq}, $\RSmi$ is shown complete for $\Smi$ functions 
$f\of\gamma_1\to\gamma_0$ that are TD polynomial-cost where $\gamma_1$ 
is an hereditarily sequential type (i.e., involves no branching data types).

  \item 
Next \S\ref{S:serial}  shows that each DP polynomial-cost 
$\Smi$-computable  $f\of\gamma_1\to\gamma_0$ can be factored as
\begin{align} \label{e:fact}
   &  f = \deserialize_{\gamma_0}\circ \widehat{f} \circ \serialize_{\gamma_1}
\end{align}
where:
\begin{inparaenum}[\em(i)]
  \item 
    $\serialize_{\gamma_1}$  DP polynomial-cost maps $\gamma_1$-values 
    to hereditarily sequential forms, 
  \item 
  	$\deserialize_{\gamma_0}$ TD polynomial-cost maps 
  	hereditarily sequential forms to type-$\gamma_0$ values, 
  \item 
  	$\widehat{f}$ is TD polynomial-cost $\Smi$-computable function
  	over hereditarily sequential types.
\end{inparaenum}
  \item 
Finally, \S\ref{S:completing} introduces $\RSmii$ which 
consists of $\RSmi$ with the addition of a \emph{compressed size 
function}, $\cs_\delta$, (Definition~\ref{d:sizes}) for each inductive
data-type $\delta$, i.e., for each $\delta$-value $v$, $\cs_\delta(v)$ is 
size of the version of $v$ with maximal structure sharing (and, so, 
minimal size). 
By using a version of the classic directed acyclic graph compression
algorithm of \citep{DST:comm:sub80}, each  $\cs_\delta$ is 
DP polynomial-cost $\Smi$-computable; hence, 
$\RSmii$ preserves $\RSmi$'s polynomial-soundness.  Moreover, 
using the $\cs_\delta$-functions and the aforementioned dag compression
algorithm, we show that we can compute the $\serialize_\gamma$ functions
in $\RSmii$. That plus the factorization of \eqref{e:fact} yields 
that each DP polynomial-cost $\Smi$ function is 
$\RSmii$ computable.
\end{itemize}

\subsection*{Related work}

Our work here is a response to the papers
\citep{gen:rem:2010,Avanzini2018OnSM}.   It is also an offshoot 
our yet to be released \citep{DR:23} where we study feasible
structural recursions and corecursions. 

Implicit computational complexity theory began with the 
work of \citet{BellantoniCook} and  \citet{Leivant:FM2} and the
ideas from these papers still exert a strong influence on 
current work, including this paper.

The work of Neil Jones, especially 
\citep{Jones93,BenAmram-Jones-2000} and \citep{Jones:book}
 strongly shaped our approach to the problems considered below.


\section{Semantic preliminaries} 

$\SetCat$, the category of sets and total functions, suffices as 
the semantic setting of this paper's programming formalisms.
Types are thus interpreted as sets and the type constructors:
product ($\times$),
coproduct ($+$), and
exponentiation ($\to$) 
have their usual $\SetCat$-interpretations and are right associative.
Let
$\pi_1\of A_1\times A_2 \to A_1$ and
$\pi_2\of A_1\times A_2 \to A_2$ 
be the canonical
product projections
and 
$\iota_1\of A_1\to A_1+A_2$ 
and $\iota_2\of A_2\to A_1+A_2$
be the canonical coproduct injections.
Also let $()$ be a $\SetCat$-constant denoting the $0$-tuple.

\emph{Notation:} 
For $1\leq j\leq \ell$, let $\inj_{j,\ell}\of\, A_j\to (A_1+\dots+A_\ell)$ 
be the canonical injection of $A_j$ into $A_1+\dots+A_\ell$
as given by:
$\inj_{\ell,\ell} = \inj_2^{(\ell-1)}$
and
$\inj_{j,\ell} = \inj_1\circ \inj_2^{(j-1)}$ when $ j<\ell$.
 

\paragraph{Polynomial functors} 
$\Id$ is the \emph{identity functor}; 
$\mathsf{C}_A$ is the $A$-\emph{constant functor} for a given set
$A$; and $F_1\times F_2$ and $F_1+F_2$
respectively denote the \emph{product} and \emph{coproduct} of
functors $F_1$ and $F_2$.  These act as Figure~\ref{fig:poly}.
A \emph{polynomial functor} $P$ is a $\SetCat$-endofunctor
inductively built from $\Id$, constant functors, products, and
coproducts;\footnote{%
  \citet{jacobs:17} calls these \emph{simple polynomial functors}.
  Grander notions polynomial functor include, for example, 
  the exponential functor, $X\mapsto X^A$.  See 
  \citep{jacobs:17}.
} 
thus grammatically:
\begin{align*}
 P \;\is\; \Id \synsep \mathsf{C}_{A} \synsep P_1\times P_2 \synsep P_1+P_2
\end{align*}
The constant-objects in our polynomial functors will always be
interpretations of types. Polynomial functors are thus type constructors.
The \emph{degree} of a polynomial functor $P$ is the degree of
$P(X)$ as an ordinary polynomial over $X$.\sidenote{E.g., 
$P_\Nat(X) = \mathsf{C}_\Unit + X$ is degree~1 and 
$P_\Tree(X) = \mathsf{C}_\Unit + X\times X$ is degree~2.} 
\begin{figure}[t]
\begin{align*}
  \Id X  &= X & \Id f  &= f \\
  C_A X  &= A & C_A f  &= \id_A
  \\
  (F_1\times F_2) X  &= F_1 X\times F_2 X
  &
  ((F_1\times F_2)\, f) \,(a,b)
  &= ((F_1\, f)(a),(F_2 \,f)(b))
  \\
  (F_1+F_2)X &= F_1 X+F_2 X
  & 
  \Quad1
  ((F_1+F_2)\,f)\,(\inj_j\,a) &= \inj_j((F_j\, f)\, a)
\end{align*}
\caption{\raggedright Action of polynomial functors, where $F_1$ and $F_2$ are $\SetCat$-endofunctors,
$f\of X\to Y$, and $A$, $X$, and $Y$ are sets.}\label{fig:poly}
\end{figure}

\section{Structural recursions}\label{S:Smi}

We formalize general first-order structural-recursive function
definitions by $\Smi$, a first-order typed lambda calculus borrowed
from \citep{DR:23}.\sidenote{%
  $\Smi$ is a restriction of $S_1$ from \citep{DR:23}, a formalism
  that includes codata and structural corecursions.}
The $\Smi$-types include products, coproducts, and inductively
defined data types.  The $\Smi$-definable functions over $\Nat$
correspond to the usual primitive recursive functions over the
natural numbers.

We have two uses for $\Smi$. 

\emph{Use 1: Reference models of computation.}
We shall consider $\Smi$ under two different operational semantics,
one using the standard top-down evaluation strategy for structural
recursions and the other using the dynamic-programming strategy of
\citep{gen:rem:2010,Avanzini2018OnSM}.  We pair both of these
operational semantics with an associated notion of the cost of a
computation.  From these we obtain two notions of
``feasible/polynomial-time'' structural recursions which will be our
standards in the following.

\emph{Use 2: A base for an implicit complexity ramification.}  Our
implicit complexity formalism $\RSmi$ is a normal/safe ramification
of $\Smi$.  This ramification expands $\Smi$'s type system and 
modifies the typing of certain constructs, but $\Smi$'s (dynamic 
programming)  operational semantics is used unchanged for $\RSmi$. 
We borrow $\RSmi$ from \citep{DR:23}.\sidenote{In \citep{DR:23}
$\RSmi$ is extended to $\RSi$, a formalism for feasible recursions 
and corecursions.}
  $\RSmi$ was inspired by the $BC$ 
function algebra from \citep[\S5]{BellantoniCook}.

\subsection{$\Smi$ types}\label{S:Smi:types}

The $\Smi$-types are given by the following.
\begin{align*}
  T & \;\is\; G \synsep G\to G
  && \hbox{($\Smi$ types)}
  \\
  G & \;\is\; N 
      \synsep G+G \synsep G \times G
  && \hbox{(Ground types)}
  \\     
  N     & \;\is\; \Unit \synsep \muP 
  && \hbox{(Normal base types)}
  \\
  P & \;\is\; \Id \synsep  \mathsf{C}_N \synsep P+P \synsep P \times P
  && \hbox{(Polynomial functors over $N$)}  
\end{align*}
The $\Smi$-types, $T$, consist of ground (level-0) types $G$ and 
level-1 types over $G$, i.e., types of the form $G\to G$.
The ground types consist of normal base types $N$ and sums and products
of ground types.  Types of the form $\muP$ are called \emph{inductive data types}, or usually just \emph{data types}. 
The normal base types\sidenote{%
  We call these \emph{normal} base types
  to be consistent with the terminology for $\RSmi$  below.
}
consist of data types together with $\Unit$, the type with () as its
sole inhabitant.  Semantically, $\muP$ is the \emph{least fixed
  point} of $P$, i.e., it is the smallest set $X$ isomorphic to
$P X$.  As polynomial functors are monotone, they have least fixed
points by the Knaster-Tarski Theorem \citep{jacobs:17}.
\textbf{N.B.} Data types, such as $\lfp{t}t$ can be 
empty/uninhabited.
We exclude empty/uninhabited data types from $\Smi$ typing judgments.\sidenote{\label{fn:mt}%
   \emph{Testing for emptiness:}  
   Suppose $P$ is a polynomial functor in which each constituent 
   base type is nonempty. Let $\mathbf{0}$ = the empty base type. 
   (Extensionally, there is only one.)  Since $P$ is monotone, $\muP
   \equiv \cup_{j\geq 0} P^{(j)}(\mathbf{0})$, and hence, 
   $\muP\equiv\mathbf{0}$
   iff $P \mathbf{0} \equiv \mathbf{0}$.  Checking 
   whether $P \mathbf{0} \equiv \mathbf{0}$ is straightforward.}

\emph{Conventions:} 
  $\gamma$ \ ranges over  ground types and
  $\delta$ \ ranges over data types.
  Sometimes in place of $\muP$ we write $\lfp{t}(P t)$ and in place
  of $\mathsf{C}_\gamma$ we simply write $\gamma$.  Also, the
  ML-like definitions of types as in \S\ref{S:intro} can be treated
  as syntactic sugar for the $\muP$-style definitions.  For example,
  the inductively defined types $\Nat$ and $\Tree$ from
  \S\ref{S:intro} and $\List_\gamma$ (lists over type-$\gamma$
  items) can be defined equivalently as follows.
\begin{align*}
   \Nat\;\; &= \Zero \synsep \Succ \Of \Nat
            \Quad{4.5} = \lfp{t}(\Unit + t)
            \Quad{1.8}   = \mu(\mathsf{C}_\Unit + \Id)    
   \\
   \Tree\;  &= \Leaf \synsep \Branch \Of \Tree\times\Tree
            \Quad{0.2} = \lfp{t}(\Unit+t\times t) 
            \Quad{0.3} =  \mu(\mathsf{C}_\Unit + \Id\times\Id) 
   \\
   \List_\gamma & = \Empty \synsep \Cons \Of \gamma\times\List_\gamma    
             \Quad{1.2} = \lfp{t}(\Unit + \gamma\times t)
             \Quad{0.1} = \mu(\mathsf{C}_\Unit+\mathsf{C}_\gamma\times \Id)
\end{align*}
\emph{Convention:} For lists we shall borrow some more ML-notation
for syntactic sugar, namely, $[\,]$ for $\Empty$ and $a :: as$ for 
$\Cons\,(a,as)$.

\begin{definition}[Sequential and branching types] \label{d:sequential} 
A ground type $\gamma$ is:
\begin{asparaenum}[(a)]
  \item \label{i:seq}
    \emph{sequential} iff 
    \begin{inparaenum}[(i)]
      \item $\gamma=\Unit$; or
      \item $\gamma=\gamma_1+ \gamma_2$ or $=\gamma_1\times\gamma_2$
        with both $\gamma_1$ and $\gamma_2$  sequential; or 
      \item $\gamma = \mu P$ where the degree of $P$ is at most one;
    \end{inparaenum}
  \item \label{i:hseq}
    \emph{hereditarily sequential} iff 
    each polynomial functor occurring within $\gamma$ is 
    of degree at most one; and
  \item \label{i:branching}
    \emph{branching} iff $\gamma$ fails to be sequential. 
\end{asparaenum}
\end{definition}

\subsection{$\Smi$ syntax and typing}

\begin{figure*}[p]
\textsc{$\Smi$ typing rules}
\begin{gather*}
    \rulelabel{Id-I}
    \irule{ 
                }{\Gamma,\,x\of\gamma\entails x\of\gamma} 
	\Quad{2.5}
    \rulelabel{$\to$-I}
    \irule{ 
                \Gamma,\,x\of\gamma_1\entails e\of\gamma_0}{
        \Gamma\entails (\lam{x} e) \of\gamma_1\to\gamma_0} 
        \Quad{2.5}
    \rulelabel{$\to$-E}\;
    \irule{ 
                \Gamma\entails e_0\of\gamma_1\to\gamma_0
                \Quad{1.5}
        			\Gamma\entails e_1\of\gamma_1}{
                \Gamma\entails (e_0\;e_1)\of \gamma_0}
\\[1ex]                
   \rulelabel{$\Unit$-I}
   \irule{}{\Gamma\entails () \of \Unit}
	\Quad{2.5}
    \rulelabel{$\times$-I}
    \irule{ 
	            \Gamma\entails e_1\of\gamma_1
	            \Quad{2}
	            \Gamma\entails e_2\of\gamma_2}{
                \Gamma\entails (e_1,e_2)
                \of\gamma_1\times\gamma_2}       
	\Quad{2.5}
    \rulelabel{$\times$-E$_j$}
    \irule{ 
	            \Gamma\entails e\of\gamma_{1}\times\gamma_2}{
                \Gamma\entails (\proj_j e)\of\gamma_j} 
\\[1ex]
    \rulelabel{$+$-I$_j$}
    \irule{ 
	            \Gamma\entails e\of\gamma_{j}}{
                \Gamma\entails (\inj_j e)\of\gamma_1+\gamma_2} 
    \Quad{2.5}
    \rulelabel{$+$-E}
    \irule{
	\Gamma\entails e_0\of \gamma_1+\gamma_2
    \Quad{1.5}             
    \set{\Gamma,x_j\of\gamma_j\entails e_j\of \gamma}_{j=1,2}}%
    {\Gamma\entails (\Case e_0 \Of \,
        (\inj_1 x_1) \Rightarrow e_1 ;\;
        (\inj_2 x_2) \Rightarrow e_2)\of\gamma}
\\[1ex]
  \rulelabel{$\Cn{\muP}$-I} 
  \irule{\Gamma\entails e\of P(\muP)}%
  {\Gamma\entails (\Cn{\muP} e)\of\muP}
\Quad{2.5}
  \rulelabel{$\Dn{\muP}$-I} 
  \irule{\Gamma\entails e\of\muP}{\Gamma\entails (\Dn{\muP} e) \of P(\muP)}
\Quad{2.5}
  \rulelabel{$\fold_{\muP}$-I}
  \irule{\Gamma \entails \lam{x}e_0\of P(\gamma)\to\gamma
         \Quad{1} 
         \Gamma\entails e_1\of\muP}%
        {\Gamma\entails \fold_{\muP}\; (\lam{x}e_0) \;e_1\of \gamma}   
\end{gather*}

\medskip
\textsc{$\Smi$ top-down evaluation rules}
\begin{gather*}
  \rulelabel{Env} 
  \irule{}{%
        x\theta\yields \theta(x)}  
\Quad{2.5}
  \rulelabel{$\lambda$-App} 
  \irule{
         e_1\theta\yields v_1\Quad{1}
         e_0\theta_0[x\mapsto v_1]\yields v 
         }{%
         ((\lam{x}e_0)\;e_1)\theta\yields v}\\[-5ex]
\\[1ex]
  \rulelabel{Unit} 
  \irule{}{()\theta \yields \underline{()}}
\Quad{2.5}
  \rulelabel{Pair} 
  \irule{e_1\theta\yields v_1 \Quad{1.25}
         e_2\theta_1\yields v_2}{%
         (e_1,e_2)\theta \yields \upair{v_1}{v_2}}
\Quad{2.}                
  \rulelabel{Proj$_j$} 
  \irule{e\theta\yields \upair{v_1}{v_2}}{%
         (\proj_j e)\theta \yields v_j}
\\[1ex]
  \rulelabel{Inj$_j$} 
  \irule{e\theta\yields v}{%
         (\inj_j e)\theta\yields (\uinj_j v)}
  \Quad{2.5}
  \rulelabel{Case} 
  \irule{e_0\theta\yields (\uinj_j v_j)\Quad{1.25} 
         e_j\theta[x_j\mapsto v_j]\yields v}{%
         (\Case e_0  \Of  \,
        (\inj_1 x_1) \Rightarrow e_1 ;\;
        (\inj_2 x_2) \Rightarrow e_2)\theta\yields
         v }
\\[1ex]
  \rulelabel{Const$_{\mu P}$}
  \irule{e\theta \yields v}%
  {(\Cn{\mu P} e)\theta \yields (\uCn{\muP} v)} 
  \Quad{2.5}
  \rulelabel{Destr$_{\mu P}$}
  \irule{e\theta \yields (\uCn{\mu P}\,v) }%
        {(\Dn{\mu P} e)\theta \yields v } 
  \Quad{2.5}
    \rulelabel{Fold$_{\mu P}$}
  \irule{f(\,g(\Dn{\mu P} e)\,) 
      \theta\yields v }%
  {(\fold_{\mu P} f\,e)\theta \yields v}
  \sidecond{\star}
  \\[1ex]
  (*)\;\parbox[t]{0.38\textwidth}{$(P\,(\fold_{\mu P} f))$ simplifies to $S$-term $g$ per the polynomial functor reduction rules. }\\[-1.5cm]
\end{gather*}
\caption{$\Smi$ typing and evaluation rules. Above $j\in\set{1,2}$
and each ground type is required to be inhabited.} 
\label{fig:Smi:typing:eval}
\end{figure*}

$\Smi$ has the following fairly standard raw syntax.
\begin{align*}
  E \;&\is\;
        X
        \synsep (\lam{X}E)
        \synsep (E_1 \; E_2)
        && \hbox{($\lambda$-calculus)}
	    \\
	    &\Quad1
	    \synsep ()
        \synsep (E_1,E_2)
        \synsep (\proj_1 E)
        \synsep (\proj_2 E)
        && \hbox{(products)}
		\\
        &\Quad1
        \synsep (\inj_1 E)
        \synsep (\inj_2 E)
        \synsep (\Case E_0  \Of  \,
        (\inj_1 X_1) \Rightarrow E_1 ;\;
        (\inj_2 X_2) \Rightarrow E_2) 
        && \hbox{(coproducts)}
        \\
        &\Quad1
        \synsep (\Cn{\muP} E)
        \synsep (\Dn{\muP} E)
        \synsep (\fold_{\muP} E_0 \,E_1)
        && \hbox{(data)}
        \\
        X & \;\is\; \hbox{identifiers}
\end{align*}
Figure~\ref{fig:Smi:typing:eval} provides the typing rules for
$\Smi$. 
These are also fairly standard, but the constructs related
to data types need some discussion.  Let $\delta = \muP$.  For
type-$\delta$ data, the constructor function
$\Cn\delta \of P(\delta)\to\delta$ and the destructor function
$\Dn\delta \of \delta\to P(\delta)$ together witness the isomorphism
between type-$P(\delta)$ data and type-$\delta$ data.  A desugaring
of the constructors of the ML-like definitions of $\Nat$ and $\Tree$
gives $\Zero = \Cn\Nat(\inj_1 ())$,
$\Succ = \lam{x} \Cn\Nat(\inj_2 x)$, $\Leaf =\Cn\Tree(\inj_1())$,
and $\Branch = \lam{(x,y)}\Cn\Tree(\,(x,y)\,)$.
The recursor $\fold_\delta$ needs to satisfy
\begin{align*} 
  (\fold_\delta f)\circ \Cn\delta & = f \circ P(\fold_\delta f), \;
  \hbox{ for all $f\of P(\gamma)\to\gamma$.}
\end{align*}
This equation expresses structural (\emph{primitive}) recursion over
$\delta$ \citep{Gibbons2002:Calculating}.
For example, 
$\plus \of \Nat\times\Nat \to\Nat$, that adds two $\Nat$s, is given by:
\begin{align*}
  &\lam{z} \left(\strut\fold_\Nat \, \left(\lam{w}\Case w \Of\, (\inj_1 w_1) \Rightarrow (\pi_1\,z); \; 
   (\inj_2 w_2) \Rightarrow \Cn\Nat(\inj_2 w) \right)\; (\pi_2\, z) \right)
\end{align*}
Definitions like this are  hard to read.  Thus, we will often give 
function definitions slathered with syntactic sugar, e.g., for $\plus$:
\begin{align*}
  &\lam{(x,y)} (\fold_\Nat \, f\; y) \quad \hbox{where  } f(\inj_1 w_1) = x; 
  \;\; f(\inj_2 w_2) = \Succ(w_2)
\end{align*}
or even restate a $\fold$ expression in terms of the equivalent 
structural recursive equations, e.g.:
\begin{align*}
  & \plus(x,\Zero) \Quad{1.4} = x\\
  & \plus(x,\Succ(y)) = \Succ(\plus(x,y))
\end{align*}

\subsection{$\Smi$  evaluation semantics}

\paragraph{Top-down  evaluation semantics}
Figure~\ref{fig:Smi:typing:eval} defines the evaluation relation,
$\yields$, that provides an operational semantics for $\Smi$ under
the top-down (non-dynamic-programming) evaluation
strategy.
\emph{Terminology:} An \emph{evaluation relation} relates closures
to values.  A \emph{closure}~$(\Gamma\entails e\of\tau)\theta$
consists of $\Gamma\entails e\of\tau$, a type judgment, and $\theta$, an
environment~for type context $\Gamma$.  (We write $e\theta$ for
$(\Gamma\entails e\of\tau)\theta$ when $e$'s typing is understood.)
$\theta$ is an \emph{environment for $\Gamma$} (notation:
$\theta\of\Gamma$) when $\theta\of \dom(\Gamma)\to\hbox{Values}$ is
such that, for each $x\in\dom(\Gamma)$, $\theta(x)$ is a
type-$\Gamma(x)$ value. 
A \emph{value}  is a \emph{value term graph} (see \S\ref{S:vtg}  below)
that represents a ground-type closure $e\theta$ where $e$ is
normal form.

Note that for simplicity, the \emph{Fold}$_\muP$ rule handles the
$(P\,(\fold_{\muP}\, f))\leadsto g$ reduction off-stage.  The rules
for this reduction are just a recasting of the functional equations
of Figure~\ref{fig:poly}.  That is,
$(\Id\, f) \leadsto f$,
$(C_A\, f) \leadsto \lam{x}x$, 
$((F_1\times F_2)\, f)\leadsto \lam{x}(((F_1\, f)(\pi_1(x)),(F_2 \,f)(\pi_2(x))))$, and 
$((F_1+F_2)\,f)  \leadsto \lam{x}(\Case x \Of \,
(\inj_1 x_1) \Rightarrow (\inj_1 ((F_1\,f)\,x_1));
(\inj_2 x_2) \Rightarrow (\inj_2 ((F_2\,f)\,x_2)))$. These rules amount to an
evaluation of $(P\,f)$ as an ordinary polynomial. 

\begin{example}\label{ex:g}
Consider $\List_{\gamma} = \muP$ where $P\, X =
\Unit+\gamma\times X$.  Then for 
$\fold_{\List_{\gamma}}$, we have $(P\; (\fold_{\List_{\gamma}} (\lam{z}e_0))
\leadsto g$ where
\begin{align*} 
g = \lam{w}\lefteqn{\Case w \Of\,}\\
         &\;(\inj_1 w_1) \Rightarrow \inj_1((\lam{x}x)\,w_1);
         \\ \nonumber
         &\;(\inj_2 w_2) \Rightarrow \inj_2\left(\left(\lam{x}\left(((\lam{y}y)\,(\proj_1 x)),\
                                           (\fold_{\List_{\gamma}} (\lam{z}e_0)\, (\proj_2 x))\right)\right)
             \; w_2\right).
\end{align*}
\end{example}

\paragraph{Dynamic programming evaluation semantics} 
In this paper we shall avoid dealing with the formal details of
the dynamic programming evaluation semantics for $\Smi$.  For these
see either \citep{DR:23} or \citep{Avanzini2018OnSM,gen:rem:2010}.


\subsection{The structure of values}\label{S:vtg}

\emph{Value term graphs} 
(or simply, \emph{values})
are particular labeled 
rooted-dags;  $v$, with and without decorations,
ranges over value term graphs.  Value term graphs 
(e.g., see Figure~\ref{fig:grow1}) 
are of the form:
\begin{itemize}
  \item 
    $\uep$ as in the \emph{Unit}-rule 
    consists of a single vertex labeled $\uep^\Unit$;
  \item 
    $(\uinj_j v)$ as in the $\mathit{Inj}_j$-rule 
    consists of  $v$ together with 
    a separate root vertex labeled by 
    $\uinj_j^{\gamma_1+\gamma_2}$ that has an out-edge to $v$'s root;
  \item 
    $(\uCn\muP v)$ as in the $\mathit{Const}_\muP$-rule 
    consists of $v$ together with 
    a separate root vertex labeled by 
    $\uCn\muP$ that has an out-edge to $v$'s root; and
  \item 
  	$\upair{v_1}{v_2}$ as in the $\mathit{Pair}$-rule
	consists of the union of $v_1$ and $v_2$ 
	together with a separate root vertex labeled by 
    $\underline{(,\!)}^{\gamma_1\times\gamma_2}$ 
    that, for $i=1,2$, has an out-edge 
    (labeled by $\proj_i$) to $v_i$'s root. 
\end{itemize}

\emph{Example:} The value term graph produced by evaluating
$\grow(\numeral1)$ is shown in Figure~\ref{fig:grow1}.
\begin{marginfigure}[1cm]
\setlength{\unitlength}{3.2pt}
\hfill\begin{picture}(15,80)(-7.5,6)
\thicklines \setstatediam{8}\setrepeatedstatediam{8}
\letstate C1=(0,85) \drawstate(C1){$\uCn{\mkern2mu\Tree}$}
\letstate C2=(0,70) \drawstate(C2){$\uinj_2$}
\letstate C3=(0,55) \drawstate(C3){$\underline{(,)}$}
\letstate C4=(0,40) \drawstate(C4){$\uCn{\mkern2mu\Tree}$}
\letstate C5=(0,25) \drawstate(C5){$\uinj_1$}
\letstate C6=(0,10) \drawstate(C6){$\uep$}
\drawtrans(C1,C2){}
\drawtrans(C2,C3){}
  \setprofcurve{-2.5} 
  \drawcurvedtrans[r](C3,C4){\scriptsize$\pi_1$ }
  \setprofcurve{2.5} 
  \drawcurvedtrans(C3,C4){\scriptsize$\pi_2$}
\drawtrans(C4,C5){}
\drawtrans(C5,C6){}
\end{picture}\hfill{\ }
\caption{$\grow(\numeral1)$'s value.}\label{fig:grow1}
\end{marginfigure}
The type superscripts on vertices help make value term graphs
explicitly typed, but we typically omit writing these superscripts
unless they are needed for clarity.

\begin{definition}[Bisimilarity] \label{d:bisim}
Two value term graphs, $v$ and $\hat{v}$, are
\emph{bisimilar}, written $v\sim \hat{v}$, iff $v$ and $\hat{v}$
have the same type $\gamma$ and:
\begin{asparaitem}
  \item 
  	$\gamma=\Unit$ (so $v=\uep=\hat{v}$); or
  \item 
  	$\gamma=\gamma_1\times\gamma_2$,
  	$v=\upair{v_1}{v_2}$, $\hat{v} = \upair{\hat{v}_1}{\hat{v}_2}$, 
  	$v_1\sim \hat{v}_1$, and $v_2\sim\hat{v}_2$; or 
  \item 
  	$\gamma=\gamma_1+\gamma_2$,
  	$v=(\uinj_i v')$, $\hat{v} = (\uinj_j \hat{v}')$, $i=j$, and 
  	$v'\sim \hat{v}'$; or
  \item 
  	$\gamma=\muP$, $v=(\uCn\muP v')$, $\hat{v}=(\uCn\muP \hat{v}')$, 
  	and $v'\sim \hat{v}'$.   
\end{asparaitem}
\end{definition}

It is clear that any $\Smi$ function maps bisimilar inputs 
to bisimilar outputs.  Hence, $\Smi$ cannot distinguish 
distinct but bisimilar values.\sidenote[][0.5cm]{The 
denotation of a ground-type $\Smi$ term is a bisimilarity equivalence class.}

\begin{definition}   \label{d:size}
  For each value term graph $v$, let $\size(v)=$
  the number of data-type constructor vertices in $v$, i.e., 
  we count the number of vertices in $v$ with 
  $\uCn\muP$-labels, but \emph{not} those with 
  $\uep$-, $\uinj_j$-, and $\underline{(,)}$-labels.
\end{definition}

By a straightforward induction on the structure of $\gamma$ 
one can show: 

\begin{lemma}\label{l:size:bnd}
  For each $\gamma$, there is a constant $k_\gamma$ such that
  for each type-$\gamma$ $v$, we have
  (the total number of 
  vertices in $v$) $\leq k_\gamma\cdot(1+\size(v))$. 
\end{lemma}

\begin{definition}\label{d:sizes} 
  For a value term graph $v$, the \emph{tree size} of $v$ (written: $\ts(v)$) and 
the \emph{compressed size} of $v$ (written: $\cs(v)$) are given by:
\begin{align*}
  &\ts(v)= \max\set{\size(v') \suchthat v \sim v'}.\\
  &\cs(v)= \min\set{\size(v') \suchthat v \sim v'}.
\end{align*}
For each  $\gamma$, let 
$\ts_\gamma$ and $\cs_\gamma$ be the restriction 
  to type $\gamma$-values of $\ts$ and $\cs$, respectively,
\end{definition}

For each $\gamma$,
it is clear that $\ts_\gamma$ is $\Smi$-computable
(as a $(\gamma\to\Nat)$-function), but a bit
less obviously, $\cs_\gamma$ is feasibly $\Smi$-computable.
We shall return to this point in \S\ref{S:completing} below. 

\subsection{The cost of $\Smi$ evaluation}

We are concerned with two operational semantics for $\Smi$:
the top-down (abbreviated, TD)
strategy of Figure~\ref{fig:Smi:typing:eval} and 
the dynamic-programming (abbreviated, DP)
strategy formalized in 
\citep{Avanzini2018OnSM,DR:23}.  Our notion of the cost of an evaluation
under either of these semantics is based on the size of the evaluation's derivation tree.\sidenote{If you prefer machines and step counting for your notions of cost,  see \S\ref{S:seq} below.} 

\begin{definition} \label{d:cost}
Suppose $\BX$ is either TD or DP. 
\begin{asparaenum}[(a)]
  \item \label{i:cost}
  	Suppose $e\theta\yields v$ under the $\BX$ evaluation strategy for 
  	$\Smi$. 
  	The $\BX$-\emph{cost} of this evaluation (written: 
	$\cost_\BX(e\theta)$)
  	is the  number of nodes in the derivation tree for 
  	$e\theta\yields v$.  

  \item \label{i:mono}
    Suppose $\entails_{S^-} f\of\gamma_1\to\gamma_0$.
    We say that $f$ is an $\BX$-\emph{polynomial-cost} 
    (abbreviated, $\BX$-\emph{poly-cost}) function iff there is a 
    polynomial function $q(\cdot)$ such that, for all 
    $\theta\of (x_1\of\gamma_1)$, \;
    $\cost_\BX(\,(f\,x_1)\theta\,)\leq q(\size(\theta(x_1)))$.
  \end{asparaenum}
\end{definition}

\begin{lemma} \label{l:costs}   
Suppose $\entails_{\Smi} f\of\gamma_1\to\gamma_0$. 
\begin{asparaenum}[(a)] 
  \item \label{i:cost:size} 
    $\size((f\,x_1)\theta\hbox{'s value})\leq \size(\theta(x_1)) + \cost_\BX((f\,x_1)\theta)$ 
    for all $\theta\of(x_1\of \gamma_1)$ and where $\BX$ is either 
    is either TD or DP. 
  \item \label{i:TD>DP}
  	If $f$ is  TD-poly-cost, then $f$ 
    is also DP-poly-cost.  
  \item\label{i:DP<TD:seq}
    If $f$ is a DP-poly-cost $\Smi$-function for which 
  	each $\fold_{\muP}$-expression within $f$ has $\muP$ sequential, 
  	then $f$ is also TD-poly-cost.
\end{asparaenum}
\end{lemma}
\begin{proof}
  For part~\eqref{i:cost:size} just note that 
  each $\uCn{\muP}$ vertex created in the course of an evaluation
  requires its own separate vertex in the evaluation derivation tree.
  Parts~\eqref{i:TD>DP} and~\eqref{i:DP<TD:seq} we leave to the reader. 
\end{proof}

\section{Ramified structural recursions}\label{S:RSmi}

\begin{figure*}[p]
\textsc{Revisions for types}
\begin{align*}
   G \;\; & \makebox[5.5cm][l]{$\is\; N \synsep S \synsep G+G \synsep G \times G$} 
   \hbox{(Ground types)}
   \\
   S  \; \;  & \makebox[5.5cm][l]{$\is\; \safe\Unit \synsep \safe{(\muP)}$}
    \hbox{(Safe base types)}            
\end{align*}
\textsc{Revisions for raw syntax} (where $\delta=\muP$)
\begin{align*}
  E &\;\is\; \dots
  \\
          &
        \makebox[5.4cm][l]{\Quad1$\synsep (\Cs{\delta} E)
        \synsep (\Ds{\delta} E)$}
         \hbox{(safe data)}
        \\
        &\makebox[5.4cm][l]{$\Quad1
        \synsep (\toSafe E)        
        \synsep (\toNorm E)$}        
         \hbox{(coercions)}
\end{align*}
\textsc{Revisions for typing} (where $\delta=\muP$)
\begin{gather*}
    \rulelabel{$+$-E}
    \irule{
	\Gamma\entails e_0\of \gamma_1+\gamma_2
    \Quad{1.5}             
    \set{\Gamma,x_j\of\gamma_j\entails e_j\of \gamma}_{j=1,2}}%
    {\Gamma\entails (\Case e_0 \Of \,
        (\inj_1 x_1) \Rightarrow e_1 ;\;
        (\inj_2 x_2) \Rightarrow e_2)\of\gamma}					\sidecond{\ast}
\\					
 \rulelabel{$\Cs\delta$-I} 
  \irule{\Gamma\entails e\of \safe{(P\delta)}}%
        {\Gamma\entails (\Cs\delta e)\of\safe\delta}
  \Quad{3}
  \rulelabel{$\Ds{\delta}$-I} 
  \irule{\Gamma\entails e\of\safe\delta}%
        {\Gamma\entails (\Ds{\delta}\,e) \of \safe{(P\delta)}}  
\\
  \rulelabel{$\fold_{\delta}$-I}
  \irule{\Gamma \entails \lam{z}e_0\of P(\safe\gamma)\to\safe\gamma
         \Quad{1} 
         \Gamma\entails e_1\of\delta}%
        {\Gamma\entails \fold_{\delta}\, (\lam{z}e_0) \;e_1\of \safe\gamma}
        \sidecond{\dagger}         
\\
  \rulelabel{toSafe-I}
  \irule{\Gamma\entails e\of\gamma}{
  \Gamma\entails (\toSafe\, e)\of\safe\gamma}
  \Quad{2.5}
  \rulelabel{toNorm-I}
  \irule{\Gamma\entails e\of\gamma}{
  \Gamma\entails (\toNorm\, e)\of\norm\gamma}
  \sidecond{
            \ddagger}
  \\[1ex]
  \text{($\ast$)\;
       $\gamma_1+\gamma_2$ is normal or $\gamma$ is safe.}
  \\
  \text{($\dagger$)\;
       $\delta$ is normal.}
  \\     
  (\ddagger)\;\,
    \parbox[t]{0.55\textwidth}{\raggedright 
    for each $x\in\fv(e)$, $\Gamma(x)$ is normal}
\end{gather*}

\smallskip
\textsc{Revisions for evaluation semantics} (where $\delta=\muP$)
\begin{gather*}
  \rulelabel{Const$_{\safe\delta}$}
  \irule{e\theta \yields v}%
  {(\Cs\delta e)\theta \yields (\uCn{\delta} v)} 
  \Quad{2.5}
  \rulelabel{Destr$_{\safe\delta}$}
  \irule{e\theta \yields (\uCn\delta\,v) }%
        {(\Ds\delta e)\theta \yields v } 
\\[1ex]
  \rulelabel{ToSafe}
  \irule{e\theta\yields v }{
   (\toSafe e)\theta\yields v }
  \Quad{2.5}
  \rulelabel{ToNorm}
  \irule{e\theta\yields v }{
   (\toNorm e)\theta\yields v }
\end{gather*} 
\caption{Revisions of $\Smi$ for $\RSmi$.}\label{f:Rsmi:rev}
\end{figure*}

$\RSmi$ is our ramified version of $S^-_1$,
where the changes to $S^-_1$ are given in Figure~\ref{f:Rsmi:rev}.
The aim of this
ramification is to forbid  the infeasible
recursions of $S^-_1$.  $\RSmi$ uses a version of Bellantoni and Cook's
normal/safe distinction that splits base-type data into two sorts:
\emph{normal} data that can drive recursions and \emph{safe} data
that is the object of recursions.  For example, in $(\fold_\delta g\,x)$
we want $x$'s type to be normal and $g\of$(normal and safe data) $\to$
(safe data).  Typing constraints enforce this normal/safe distinction, 
which is
roughly the idea behind the $BC$ function algebra of
\citep[\S5]{BellantoniCook} and the formalism of
\citep{Leivant:FM2}, but \textbf{not} the better known $B$ function
algebra \citep[\S2]{BellantoniCook}.  {The normal/safe
  distinction applies to just ground types.}

\subsection{$\RSmi$ types} 
Changes to $\Smi$'s types are detailed in Figure~\ref{f:Rsmi:rev}.
$\RSmi$ inherits all of $\Smi$'s normal base types.
Paralleling  each normal data type $\muP$,
we have the corresponding safe type is $\safe{(\muP)}$.
Also, $\safe\Unit$ is the safe version of $\Unit$.
We extend -$\safe{}$ to all ground types by: 
\begin{gather*}
\safe{(\sigma_1\times\sigma_2)}=\safe{\sigma_1}\times\safe{\sigma_2} 
\qquad
\safe{(\sigma_1+\sigma_2)}=\safe{\sigma_1}+\safe{\sigma_2}
\qquad
\safe{(\safe\gamma)}=\safe\gamma
\end{gather*}
Also, let $\norm\gamma$ denote the version of $\gamma$ with all 
$\mathsf{S}$'s removed.   A \emph{normal} (respectively, \emph{safe}) ground 
type
is one in which each of constituent base types is normal (respectively, safe).  A \emph{mixed type} is a
ground type, such as $\Nat\times\safe{\Lst_\Nat}$,   that is neither normal nor safe.

\subsection{$\RSmi$ syntax and typing}
$\RSmi$ inherits  $\Smi$'s raw syntax  and adds safe data-type
constructors and two type-coercion operators:
$\toSafe$ and $\toNorm$.  (See Figure~\ref{f:Rsmi:rev}.)
Let $\delta=\muP$.  
The type $\safe\delta$ has
the constructor
$\Cs\delta\of\safe{(P \delta)}\to\safe\delta$ and destructor
${\Ds\delta}\of\safe\delta\to\safe{(P \delta)}$. 
We sometimes use
sugared constructors for $\safe\delta$, e.g., $\safe\Succ \of
\safe\Nat\to\safe\Nat$.
A type such as 
$\safe{(\Lst_\Nat)}$ is usually written  as
$\safe{\Lst_\Nat}$.%
%
%
%
\footnote{\textbf{N.B.} Inductive data types \emph{cannot} mix normal 
and safe things,
  e.g., a normal list of $\Nat$s is fine, as is a safe list of 
  $\safe\Nat$s, but not a normal list of $\safe\Nat$s nor 
  a safe list of $\Nat$s.}
Also, $\safe{()}$ is the sole inhabitant of $\safe\Unit$.

\begin{figure}[t]
\begin{leftbar}\begin{align*}
  & \plus' \of \safe\Nat \times \Nat \to\safe\Nat
  \\
  & \plus' (x,y) = \fold_\Nat \incr \;y 
  \\
  & \Quad2 \hbox{where } \incr \of (\Unit + \safe\Nat)\to\safe\Nat
  \\
  & \Quad{5.1}  \incr (\inj_1 ()) = x
  \\
  & \Quad{5.1}  \incr (\inj_2(n)) = \safe\Succ(n)
  \\[1ex]
  &\mathit{times}' \of\Nat\times\Nat \to\safe\Nat
  \\
  & \mathit{times}'(x,y) = \fold_\Nat f \;y
  \\
  &\Quad{2}\hbox{where } \step \of (\Unit + \safe\Nat)\to\safe\Nat
  \\
  & \Quad{5.1} \step(\inj_1()) = \safe\Zero;  
  \\
  & \Quad{5.1} \step(\inj_2(n)) = \plus'(n,x)
  \\[1ex]
  &\plus, \,\mathit{times} \of\Nat\times\Nat \to\Nat
  \\
  & \plus (x,y) = \toNorm(\plus'(\toSafe(x),y))
  \\
  & \mathit{times}(x,y) = \toNorm(\mathit{times}'(x,y))
  \\[1ex]
  &\mathit{sumLst} \of {\Lst_\Nat} \to\Nat
  \\
  &\mathit{sumLst}(xs) 
   = \toNorm(\fold_{\Lst_\Nat} g\; \xs)
   \\
  & \Quad{2}\hbox{where } g\of(\Unit + \Nat\times\safe\Nat) \to\safe\Nat
  \\
  & \Quad{5.1} g(\inj_1 ()) = \safe\Zero
  \\
  & \Quad{5.1} g(\inj_2 (x,s)) = \plus' (s,x)
\end{align*}\end{leftbar}
\caption{Sample $\RSmi$ recursions \\[1ex]
 \textbf{\emph{Important convention:}} Displayed 
$\RSmi$  definitions have a 
left-margin bar; displayed $\Smi$ definitions 
do not.
}
\label{f:RSmi:ex}
\end{figure}

$\RSmi$ revises  $\Smi$'s typing rules by adding
new side-conditions to the $+$-\emph{E} and $\fold_\delta$-\emph{I} rules
and adding rules for the safe data-type
constructors and type-coercion operators.  (See Figure~\ref{f:Rsmi:rev}.)
The new side-condition on $+$-E helps 
make the value of a normal-type expression independent the 
values of its safe-type subexpressions.\sidenote{See \S\ref{S:noninterf}.}  
The new side condition for $\fold_{\delta}$-\emph{I} is key to reducing
the power of folds down to mere polynomial-time. 
Figure~\ref{f:RSmi:ex} provides some sample $\RSmi$ $\fold$-recursions.

The $\toSafe$ operator simply shifts the type $\gamma$ assigned to 
the value of an expression to the type's safe version, $\safe\gamma$.
The $\toNorm$ operator does a safe-to-normal shift, but under
a  strong side-condition. 
The $\toNorm$ construct 
addresses a perennial difficulty of ramified type systems: 
that certain sensible compositions can be untypable.
E.g., for $\mathit{times}'$ as in
Figure~\ref{f:RSmi:ex},
$\mathit{cube} = \lam{x}(\mathit{times}'\,x\,(\mathit{times}'\,x\,x))$ 
fails to type in $\RSmi$. The use of $\toNorm$ in $\mathit{times}$
in Figure~\ref{f:RSmi:ex}
mitigates this difficulty.\sidenote{%
  The $\toNorm$-I rule is an adaptation to $\lambda$-calculi of 
  Bellantoni and Cook's \citeyearpar{BellantoniCook} \emph{Raising Rule} 
  for their BC formalism which in turn is an adaptation of 
  their \emph{safe composition} scheme for the $B$ formalism of 
  the same paper.}   
By convention: 
$\toSafe(\,()\,)$ = $\safe{()}$ and
$\toNorm(\,\safe{()}\,) = ()$.

\subsection{The operational semantics of~$\RSmi$}
If we exclude $\toNorm$ and $\toSafe$, then
$\RSmi$ and $S^-$ share the same DP-operational semantics which 
is oblivious to  normal/safe distinctions. 
(Evaluating a $\Cs\delta$ constructor thus builds a $\uCn\delta$-vertex.)
The  normal/safe  restrictions 
constrain what programs can be constructed, but
these type restrictions 
have no influence on how well-typed programs behave.
Consequently, the 
operational semantics of $\toNorm$ and $\toSafe$ are taken to be trivial:
they simply return their argument with no new
vertices being created. 

\begin{definition}
For a $\RSmi$ type $\gamma$, we say that 
a value-term graph $v$ is an $\RSmi$-type-$\gamma$ value if and only if
$v$ is a $S^-$-type-$\norm\gamma$ value. 
\end{definition}

\subsection{Polynomial-time soundness and seeming incompleteness}

\begin{theorem}[Polynomial-time soundness for $\RSmi$]\label{t:rsmi:sound:1}  
   Suppose that $\entails_{\RSmi} f\of\gamma_1\to\gamma_0$ where 
   $\gamma_1$ is normal.  Then $f$ is DP-poly-cost. 
\end{theorem}
      
This theorem follows from more general results in \citep{DR:23}.  
However, as that paper is currently not near public release,
in this version of the paper we have included an 
appendix (\S\ref{S:soundness}) that sketches  the basic soundness 
results for 
$\RSmi$.

Our central concern in this paper is not soundness, but completeness.
Let us recall the $\Smi$-definition of the DP-poly-cost function
$\height \of\Tree\to\Nat$ which we suspect is
not $\RSmi$-computable:
\marginnote[6ex]{\textbf{\emph{Reminder:}} Displayed 
$\RSmi$  definitions have a 
left-margin bar; displayed $\Smi$ definitions 
do not.}
\begin{align*}
  &\height \,\Leaf = \numeral0
  \\
  &\height \,(\Branch(tl,tr)) = \Succ(\Max(\height\,tl,\height\,tr))
\end{align*} 
For an $\RSmi$-definition along the lines of the above, 
both expressions $(\height\,tl)$ and $(\height\,tr)$ 
would necessarily be of type $\safe\Nat$, but 
we would need at least one of $\Max$'s arguments to be normal.
If somehow we had a normal $\Nat$ larger than both 
$(\height\,tl)$ and $(\height\,tr)$, we could compute the 
required maximum, but from whence would this normal argument come?
The formalisms of \citep{gen:rem:2010,Avanzini2018OnSM} 
seems to have similar problems.

\paragraph{Why the $\cs_\gamma$ functions are good candidates 
for $\RSmi$'s missing ingredient}
Consider  $\entails_{\Smi} f\of\gamma\to\Nat$ where $f$ is DP-poly-cost and $\theta\of(x\of\gamma)$.  
For such a term we know, by Lemma~\ref{l:costs}\eqref{i:cost:size}, that there is a polynomial
$p$ such that, for all $\theta$, $\size((f\,x)\theta$'s value$) 
\leq p(\size(\theta(x)))$.  But since $f$ must produce the same $\Nat$ for all bisimilar
inputs, this inequality is really 
\begin{align*}
  \size((f\,x)\theta\hbox{'s value}) \leq p(\cs_\gamma(\theta(x))).
\end{align*}  
The $\cs_\gamma$'s turn out to be DP-poly-cost $\Smi$-computable
(Corollary~\ref{c:cs} below), but
when we try to use $\RSmi$ to  compute, say, $\cs_{\Tree}$, 
there appears the same sort of problem encountered in trying to compute $\height\of\Tree\to\Nat$.  Moreover, computing $\height\of\Tree\to\Nat$
in $\RSmi$ would be straightforward \emph{if} we had access to the 
$\cs_\gamma$'s.

\subsection{Comparing formalisms}\label{S:compare}

Here we briefly compare $\RSmi$ to the  formalism
of \citep{Avanzini2018OnSM}, which  we shall call ADL.\sidenote{ADL
amounts to an extension of the formalism of \citep{gen:rem:2010}.}
Despite surface differences, 
ADL and $\RSmi$ are broadly comparable systems in most respects.
ADL is a function algebra using multiple tiers, $\RSmi$ is a
$\lambda$-calculus using just two tiers, but for type-level-1 systems
these differences are largely a matter of taste. ADL includes 
simultaneous structural recursions, but $\RSmi$ can mimic these 
by use of product and coproduct types.\sidenote{E.g., If $f_1\of\gamma\to\gamma_1$ and $f_2\of\gamma\to\gamma_2$ are 
defined via simultaneous structural recursion, in $\RSmi$ one 
can define $f\of\gamma\to(\gamma_1\times\gamma_2)$ such that
$f_1=\proj_1\circ f$ and
$f_2=\proj_2\circ f$.}  
We do claim, however, that
inductively defined data types are a bit more developed in $\RSmi$
than in ADL.

\section{Completeness of $\RSmi$ over hereditarily sequential types}
\label{S:seq}

The goal of this paper is to patch an apparent weakness in $\RSmi$.
To do this we first need to understand
$\RSmi$'s strengths.  The following theorem establishes that $\RSmi$
is complete with respect to a particular broad class of
feasible structural recursions.

\begin{theorem} \label{t:hs}
  Suppose $\entails_{\Smi}f\of   
  \gamma_1\to\gamma_0$ where 
  $\gamma_1$ is normal and hereditarily sequential 
  (Definition~\ref{d:sequential}(\ref{i:hseq}))
  and where $f$ is TD-poly-cost. Then $f$ 
  is $\RSmi$-computable. 
\end{theorem}

The rest of this section is devoted to proving Theorem~\ref{t:hs}. 

The first key fact to note is that, in $\RSmi$,  one can compute an upper 
bound on the size of a hereditarily sequential value. 

\begin{lemma} \label{l:hssize}
  For each normal, hereditarily-sequential $\gamma$,  
  $\ts_\gamma$ is $\RSmi$-computable. 
\end{lemma}
\begin{proofsketch}  
For each normal, hereditarily sequential $\gamma$, we introduce a
closed  $\RSmi$-function
$\treeSize_\gamma\of\gamma\to\Nat$  that computes $\ts_\gamma$. 
The $\treeSize_\gamma$ functions are defined inductively on the structure of $\gamma$.
Recall that Figure~\ref{f:RSmi:ex} introduced $\RSmi$-functions
$\plus$ and $\plus'$ such that  
$\plus(\numeral{m_0},{\numeral{m_1}}) = {\numeral{m_0+m_1}}$
and $\plus'(\toSafe(\numeral{m_0}),{\numeral{m_1}}) = \toSafe(\numeral{m_0+m_1})$.
The $\Unit$, sum, and product cases are
 straightforward:\sidenote{\emph{Recall:} $\size$ (Definition~\ref{d:size})
 does not count $\uep$-, $\uinj_j$-, or $\underline{(,)}$-vertices.}
\vspace*{-3ex}\begin{leftbar}\begin{align*}
& \treeSize_\Unit (x) \Quad{0.7} = \numeral0
\\
& \treeSize_{\gamma_1+\gamma_2}(x) = 
\Case x \Of \left\{\strut(\inj_j x_j) \Rightarrow \treeSize_{\gamma_j}(x_j)\right\}_{j=1,2}
\\
& \treeSize_{\gamma_1\times\gamma_2}(x) = \plus(
\treeSize_{\gamma_1}(\pi_1(x)) ,\treeSize_{\gamma_2}(\pi_2(x))) 
\end{align*}\end{leftbar}
Next let us consider the special case of 
$\gamma=\List_{\gamma_0} = \lfp{t}(\Unit+\gamma_0\times t)$
where the lemma is know to hold for $\gamma_0$. 
For this case we have:
\vspace*{-3ex}\begin{leftbar}\begin{align*}
  & \treeSize_{\List_{\gamma_0}}(x) = \toNorm(\fold_{\List_{\gamma_0}} \tally \; x)
  \\
  & \Quad1\hbox{where}\; \tally\of(\Unit + {\gamma_0}\times\safe\Nat) \to\safe\Nat 
  \\
  & \Quad4 \tally (\inj_1 ()) \Quad{2} = \toSafe(\numeral1) 
  \\
  & \Quad4 \tally(\inj_2 (w,n')) = \safe\Succ(\plus'  (n',\treeSize_{\gamma_0}(w)))
\end{align*}\end{leftbar}
Above, the $\tally (\inj_1 ())$ equation corresponds to the $\Empty$ case and 
the
$\tally(\inj_2 (w,n'))$ equation  corresponds to the $\Cons(w,x')$ case in which
$n'=\treeSize_{\List_{\gamma_0}}(x')$.
Thus $\treeSize_{\List_{\gamma_0}}(x)$ is the sum of:
$1$ (for the $\Empty$ constructor) and, for each item in the list, 
(the $\treeSize_{\gamma_0}$ of the item$)+1$ (where the $+1$ 
counts the $\Cons$ constructor for this item). 
The general $\gamma=\muP$ case is just the 
$\List_{\gamma_0}$-case with more bureaucracy. 
\end{proofsketch}

\begin{corollary}\label{cor:ts:bnd}
  For each hereditarily sequential $\gamma$,  there is a 
  polynomial function $q_\gamma(\cdot)$ such that, for each type-$\gamma$ 
  value $v$: \ $\ts_\gamma(v)\leq q_\gamma(\size(v))$. 
\end{corollary}

\begin{proof}
This follows from  Lemmas~\ref{l:costs}\eqref{i:cost:size}
and~\ref{l:hssize} and Theorem~\ref{t:rsmi:sound:1}.
\end{proof}

Returning to the proof of Theorem~\ref{t:hs},
suppose $f=\lam{x_1}e_0$ is as in 
the theorem's hypothesis.
Since $\lam{x_1}e_0$ is TD-poly-cost,
there is a polynomial function
$q_0(\cdot)$  such that, for all  $\theta_0\of(x_1\of\gamma_1)$, 
$\cost_{TD}(e_0\theta_0)$ is $\leq q_0(\size(\theta_0(x_1)))$ which, since
$q_0(\cdot)$ is monotone, is $\leq q_0(\ts_{\gamma_1}(\theta_0(x_1)))$.
It follows then 
from Lemma~\ref{l:hssize} that there is an $\RSmi$ expression $\bnd_0$
with $x_1\of\gamma_1\entails \bnd_0\of\Nat$ such that  
$\bnd_0\theta_0\yields \numeral{q_0(\ts_{\gamma_1}(\theta_0(x_1)))}$
for all $\theta_0\of(x_1\of\gamma_1)$. 
So, $\bnd_0$ gives us an $\RSmi$-computable upper bound on $e_0$'s TD-cost.

\begin{figure*}[ph]
\begin{align*}
\note{$\mkern-15mu\lambda$-calculus}
\\
   &(x,\theta,k)
   \TO  (\theta(x),\theta,k)
   &&\rulenum{1}
   \\
   &(((\lam{x}e_0)\;e_1),\theta,k)
   \TO  
   (e_1,\theta,(\App_1\,(\lam{x_1}e_0))::k)
   &&\rulenum{2a}
   \\
   &(v_1,\theta,(\App_1\;(\lam{x_1}e_0))::k)
   \TO  
   (e_0,\theta[x_1\mapsto v_1],\App_2::k)
   &&\rulenum{2b}   
   \\
   &(v,\theta[x_1\mapsto v_1],\App_2::k)
   \TO   
   (v,\theta,k)
   \Quad9 (*)
   &&\rulenum{2c}
\\[2ex]
 \note{$\mkern-15mu$products} 
 \\
   &((),\theta,k) 
   \TO    
   (\uep,\theta,k) 
   &&\rulenum{4}
   \\
   &((e_1,e_2),\theta,k)    
   \TO    
   (e_1,\theta,(\Pair_1\,e_2)::k)
   &&\rulenum{4a}
   \\
   &(v_1,\theta,(\Pair_1\,e_2)::k)    
   \TO    
   (e_2,\theta,(\Pair_2\,v_1)::k)
   &&\rulenum{4b}
   \\
   &(v_2,\theta,(\Pair_2\,v_1)::k)    
   \TO     
   (\upair{v_1}{v_2},\theta,k)
   &&\rulenum{4c}
   \\
   &((\proj_j e),\theta,k)    
   \TO    
   (e,\theta,(\Proj\,j)::k)
   &&\rulenum{5a}
   \\   
   &(\upair{v_1}{v_2},\theta,(\Proj\,j)::k)   
   \TO   
   (v_j,\theta,k)
   &&\rulenum{5b}
\\[2ex]
   \note{$\mkern-15mu$coproducts}
   \\
   &((\inj_j e),\theta,k) 
   \TO    
   (e,\theta,(\Inj\,j)::k)
   &&\rulenum{6a}
   \\
   &(v,\theta,(\Inj\,j)::k) 
   \TO    
  ((\uinj_j{v}),\theta,k)
   &&\rulenum{6b}
   \\
   &\lefteqn{(\Case\,e_0\,\Of (\inj_1 x_1) \Rightarrow e_1; (\inj_2 x_2) \Rightarrow e_2),
   \theta,k)}
   \TO  
   (e_0,\theta,(\mathsf{Case}_1\,x_1\,e_1\,x_2\,e_2)::k)   
   &&\rulenum{7a}
   \\
   & \lefteqn{((\uinj_j v_j),\theta,(\mathsf{Case}_1\,x_1\,e_1\,x_2\,e_2)::k)}   
   \TO   
   (e_j,\theta[x_j\mapsto v_j],\mathsf{Case}_2::k)
   &&\rulenum{7b}
   \\
   & \lefteqn{(v,\theta[x_j\mapsto v_j],\mathsf{Case}_2::k)}   
   \TO   
   (v,\theta,k)\Quad9 (\dagger)
   &&\rulenum{7c}   
\\[2ex]
   \note{$\mkern-15mu$data/recursion}
   \\
   &((\Cn{\delta} e),\theta,k) 
   \TO    
   (e,\theta,(\Const\,\delta)::k)
   &&\rulenum{8a}
   \\
   &((v,\theta,(\Const\,\delta):k) 
   \TO    
   ((\uCn{\delta} v),\theta,k)
   &&\rulenum{8b}
   \\   
   &((\Dn{\delta} e),\theta,k)
   \TO   
   (e,\theta,\Destr::k)
   &&\rulenum{9a}
   \\
   & ((\uCn\delta v),\theta,\Destr::k)
   \TO   
   (v,\theta,k)  
   &&\rulenum{9b}
   \\
   &((\fold_{\delta} f\,e),\theta,k)
  \TO  
   (f(g(\Dn\delta e)),\theta,k)\Quad5(\ddagger)    
   &&\rulenum{10}
\end{align*}
  $(*)$~~\parbox[t]{0.65\textwidth}{\raggedright
  Since $\Smi$ is a type-level~1 formalism, environments are stacklike.
  Hence, the $\theta[x_1\mapsto v_1] \leadsto \theta$ transition amounts to  a stack-pop.}\smallskip
  
  $(\dagger)$~~\parbox[t]{0.65\textwidth}{\raggedright
  As in $(*)$, $\theta[x_j\mapsto v_j] \leadsto \theta$ 
  amounts to a stack-pop.}\smallskip
  
  $(\ddagger)$~~\parbox[t]{0.65\textwidth}{\raggedright
  If $\delta=\mu P$, then $(P\,(\fold_\delta f))$ simplifies to $S_1^-$-term $g$ per the polynomial functor reduction rules.}
\caption{The $\Smi$-CEK using the top-down evaluation 
         strategy.}
\label{f:cek:s1-}
\\[1ex] 
\end{figure*}

\subsection{An abstract machine for $\Smi$} 

To help make use of the above $\RSmi$-computable upper bound, 
we provide a CEK abstract machine (see Figure~\ref{f:cek:s1-})
derived from the TD-evaluation rules of Figure~\ref{fig:Smi:typing:eval}.
CEK machines are due to \citet{FF87}.\sidenote{Also see \citep{FF03}.}  
Our $\Smi$-CEK machine
is a set of rules for evaluating  a $\Smi$-closure, $e\theta$, by doing 
a left-to-right depth-first traversal of 
$e\theta$'s evaluation derivation-tree (built per the rules of Figure~\ref{fig:Smi:typing:eval}).  
Machine \emph{states} are triples, $(c,\theta,k)$, where
$c$ is a \emph{context} which is either an
expression  or else a value, $\theta$ is an environment, and 
$k$ is a \emph{continuation}.  Here a {continuation} is a
stack of to-do items for the derivation-tree traversal. 
The machine's initial state  is $(e,\theta,[\,])$ where $e\theta$
is the closure to be evaluated and $[\,]$ is the empty stack.
The final state is of the form $(v,\theta',[\,])$ where $e\theta\yields v$.
The rules tell us how to rewrite the initial state to the final state.  
\begin{itemize}
\item
If the left-hand side of a rule has a state of the form
$(e,\theta,k)$ where $e$ is an expression, then either:
\begin{itemize}
  \item $e$ is immediately reduced to a value (rules R1, R2, and R3), or 
  \item $e$ is broken into subexpressions, one of which, $e'$, 
is picked out to be evaluated next and with a record, $r$, 
pushed onto the continuation stack where $r$ indicates what
is to be done with the value of $e'$
(Rules R2a, R4a, R5a, R6a, R7a, R8a, and R9a), or 
  \item $e$ is fold-expression which is expanded out one step
  (Rule 10).
\end{itemize}
\item
If the left-hand side of a rule has a state of the form
$(v,\theta,r::k)$ where $v$ is a value and $r$ is the record
at the top of the stack, then the record details what to do
with $v$ and $\theta$ to either: 
\begin{itemize}
  \item produce another value (rules  R4c, R5b, R6b, R7b, R8b, and R9b) 
	or 
  \item pop the environment stack (rules R2c and R7c) or
  \item evaluate another expression $e'$ while saving $v$
  for later use with the value of $e'$ (rules R2b and R4b). 
\end{itemize}
\end{itemize}

\emph{Note:}
For each vertex in the derivation tree for $e\theta$, there
are at most three steps in the CEK's execution.  Hence, in a 
CEK evaluation of our $e_0\theta_0$, the machine makes at most 
$3\cdot q_0(\ts_{\gamma_1}(\theta_0(x_1)))$ steps. 

\subsection{Implementing a CEK next-step function in $\RSmi$ to evaluate $e_0$}  

We want an $\RSmi$-implementation of the CEK specialized to evaluations 
of $e_0\theta_0$ for $\theta_0\of(x_1\of\gamma_1)$.  To do this
we need to represent CEK states with a $\Smi$/$\RSmi$-type 
$\State$ and write an $\RSmi$-function $\step \of 
(\Unit+\safe\State) \to\safe\State$  such that:
\begin{align*}
  &\step(\inj_1 ()) \;= \hbox{ the initial $\State$ of a CEK-evaluation of $e_0$.}\\ 
  &\step(\inj_2 (s)) = \begin{cases}
  s', & \hbox{if }\parbox[t]{6cm}{$s$ is not final and $s'$ is the next $\State$ after $s$ in this 
	CEK-evaluation of $e_0$;}\\
  s, & \hbox{otherwise.} 
	\end{cases}
\end{align*}
Thus from the value of 
$(\fold_\Nat\,\step\;\numeral{m})\theta_0$, we can read off the 
value of $e_0\theta_0$ provided $m\geq 3\cdot q_0(\ts_{\gamma_1}(\theta_0(x_1)))$.

To formalize $\State$ and $\step$ we need sketch how to represent
the various components of states, e.g., expressions, values, 
continuations, and environments.
Since we have branching types, it is straightforward to 
define a data-type for parse trees of $\Smi$-expressions by:
\begin{align*}
 &    \!\datatype \Term = \Var \Of \String \synsep \mathsf{Lambda}\Of \String\times \Term \synsep \dots 
\end{align*}
The continuation stack can be straightforwardly represented 
as a simple list of appropriate records:  
\begin{align*}
   & \datatype\, \Record = \mathsf{App}_1 \Of \Term \synsep \dots 
   \\
   & \mathsf{type}\, \Kont = \List_\Record
\end{align*}
We shall discuss how to 
represent
environments shortly. Let us first note that under the assumptions
that we have made thus far, all of CEK rules except R1, R7b, and R10
are shallow rearrangements of fragments of syntax and values and 
thus are straightforward to express in $\RSmi$ as a $\safe\State$ to $\safe\State$ transformation.  In regards to rule R10, in a 
$(P\,(\fold_\delta f))\leadsto g$
reduction, the number of reduction steps depends solely on $P$. 
Thus, relative to $e_0$ (which is fixed), each of these 
particular reductions can be done in the course of a safe-to-safe 
computation in carrying out rule R10.

For environments, we need to represent just those environments adequate 
for those expressions appearing in a $e_0\theta_0\yields v$ evaluation.
Let $\hat\gamma_0,\dots,\hat\gamma_{\ell-1}$ be a list of ground $\Smi$-types
that includes all the  ground types occurring in $\lam{x_1}e_0$'s type
derivation and let: 
\begin{align*}
	\mathsf{type}\, \Grd = \hat\gamma_0+\dots+\hat\gamma_{\ell-1}
\end{align*}	
Without loss of generality, we take 
$\hat\gamma_0=\gamma_0$ and
$\hat\gamma_1=\gamma_1$
For each $i<\ell$, let $\hat{c}_i$ be some constant of type $\hat\gamma_i$.
(Recall that we have forbidden empty types.)  We can extract
a $\hat{\gamma_i}$-value from a $\Grd$-value by:
\vspace*{-4ex}\begin{leftbar}\begin{align*}
  & \extract_{i} \of \safe\Grd \to \safe{\hat\gamma_i}
  \\
  & \extract_{i}\, u = \Case u \Of (\inj_{i+1,\ell} x) \Rightarrow x;\; 
  \mathsf{otherwise} \Rightarrow \toSafe(\hat{c}_i)
\end{align*}\end{leftbar}
We shall also need versions of the various $\Smi$ constructors and
destructors for $\Grd$.  For example:
\vspace*{-2ex}\begin{leftbar}\begin{align*}
  & \mathit{pair} \of \safe\Grd\times\safe\Grd \to \safe\Grd
  \\
  &\!\left\{\mathit{pair}(\inj_{i+1,\ell} v_1,\inj_{j+1,\ell} v_2) = \Strut{1cm}
  \begin{cases}
      \inj_{1,\ell}(\,\hat{c}_0\,),     
         & \hbox{if $\hat{\delta}_i\times\hat{\delta}_j$ has no 
                 $\hat{\delta}$-index;}
      \\
      \inj_{k+1,\ell}(\,(v_1,v_2)\,), 
         & \hbox{if $k$ is the least $\hat{\delta}$-index of 
                  $\hat{\delta}_i\times\hat{\delta}_j$}
  \end{cases}\right\}_{i,j< \ell}
\end{align*}\end{leftbar}
We leave it to the reader to fill in the details for the 
analogues of various other constructors and destructors.

Let $\String$ be a type for strings and 
let $\Alist=\List_{\String\times\Grd}$ which we use
to represent an environment as an association list.  To extend an
environment with a new variable binding (as in Rules R2b and R7b),
it suffices to add a $\String$-$\Grd$ pair to the front of the 
$\Alist$. 
Searching the environment (as in Rule R1) involves a 
bit more work. 
Let us assume that
$\equal \of \Nat\times\safe\String\to\safe\String\to \safe{(\Unit+\Unit)}$
is an $\RSmi$-function 
such that $\equal(\numeral{m},s,s')=(\inj_1\,\safe{()}) ~(\equiv \mathsf{true})$ if 
$s$ and $s'$ are equal $\String$s of length $\leq m$, and 
$=(\inj_2\,\safe{()}) ~(\equiv \mathsf{false})$ otherwise.  
Then we can write $\lookup$ as in Figure~\ref{fig:lookup}.
\begin{figure}[t]
\begin{leftbar}\begin{align*}
  & \lookup \of \Nat \times \safe\String \times \safe\Alist \to\safe\Grd 
  \\
  &\lookup\,(n,s,as) = \search\,(n,as)
  \\
  & \Quad2 \hbox{where } \search\of\Nat\times \safe\Alist \to\safe\Grd 
  \\
  &\Quad5   \search\,(\Zero,\,as)
  \Quad{5.8} = (\inj_{1,\ell}\,\hat{c}_0)
  \\
  &\Quad5   \search\,((Succ\,n'),[])
  \Quad{4.3} = (\inj_{1,\ell}\,\hat{c}_0)
  \\
  &\Quad5 \search\,((Succ\,n'),\,(s',u)::as') = \\
  &\Quad7
  \Let w = \search(n',\,as')\\
  &
  \Quad{7.5} \In\,\Case \equal(n,s,s') \Of\,(\inj_1 x_1)\Rightarrow u; (\inj_2 x_2) \Rightarrow w
\end{align*}\end{leftbar}\label{fig:lookup}\caption{The $\lookup$ function}
\end{figure}

Now suppose $m_0\geq q_0(\ts_{\gamma_1}(\theta_0(x_1)))+k$ where $k$ is 
the length of the longest variable name in $\lam{x_1}e_0$
(and the $g$'s in the possible rule R10 expansions).
Then it is clear that
$\lam{(s,as)}(\lookup\,(\numeral{m_0},s,as))$
will successfully look up values in the environment 
the R1-steps of the CEK-evaluation of $e_0\theta_0$. 

To finally construct our $\RSmi$-function $f'$
that computes $f$ (=$\lam{x_1}e_0$), we need two more types. 
\begin{align*}
   & \!\datatype \Context = \mathsf{E} \Of \Term 
                          \synsep \mathsf{V} \Of \Grd 
   \\
   & \mathsf{type}\, \State = \Context \times\Alist\times \Kont
\end{align*}
For the type $\Context$, the $\mathsf{E}$ constructor is for expressions,
the $\VAL$ constructor is for ground-type values.
We outline $f'$ in Figure~\ref{fig:f0'}.  Filling in the 
details of $f'$ and checking that this $f'$ suffices 
are both straightforward and left to the reader. 

\begin{figure}[t]
\begin{leftbar}
\begin{align*}
   &f' \of \gamma_1 \to \gamma_0
   \\
   &f' = \lam{x_1} \toNorm\big( \\
   & \Quad5 \Let\; (t_0\of\type{exp}) = \hbox{(an expression 
   constructing the parse tree for $e_0$)};
   \\
   & \Quad{6.5} (as_0\of\Alist) = (nx_1,(\inj_{1,\ell} x_1))::[];
   \quad \Comment{$nx_1 =$ the $\String$ name of variable $x_1$}
   \\
   & \Quad{6.5} (s_0\of\safe\State) = \toSafe(\,((\mathsf{E}\, t_0),as_0,[])\,);
   \\
   & \Quad{6.5} (b_0\of\Nat) = 
   \left(\hbox{an expression computing 
   $\numeral{3\cdot q_0(\ts_{\gamma_1}(x\hbox{'s value}))}$}\right);
   \\
   & \Quad{7.5} \vdots \Quad1 
   \hbox{\em(definitions of the various auxiliary functions as mentioned above)}
   \\[1ex]
   & \Quad{6.5} \step \of (\Unit + \safe\State) \to \safe\State               
   \\
   & \Quad{6.5} \step\,y = \Case y \Of 
   \\
   &\Quad{11}   (\inj_1\, ()) \Rightarrow s_0; \\
   &\Quad{11}   (\inj_2\, (c,e,k)) \Rightarrow 
   \\
   &\Quad{12}\Case\, (c,e,k) \Of
   \\
   & \Quad{13} ((\mathsf{E}\, (\mathsf{Var} \,nx)\,),as,k) 
               \Rightarrow ((\VAL\, (\lookup(b_0,\,nx,\,as))),as,k);
   \\
   & \Quad{13.5} \quad\vdots  \quad
   \hbox{\em(cases left to the reader)} 
   \\[1ex]
   & \Quad{13} ((\mathsf{V}\, v),as,[])
               \Rightarrow ((\mathsf{V}\, v),as,[])
  \\
   & \Quad{6.5} (c_*,e_*,k_*) =  \fold_\Nat \step \;b_0  \\
   &\Quad{5.25} \In \, \Case c_* \Of \\
   &\Quad7     (\VAL \,v) \Rightarrow \extract_1(v); \\
   &\Quad7     \mathsf{otherwise} \Rightarrow \hat{c}_0 
   \\
   &\Quad{4.5} \big)         
\end{align*}\end{leftbar}
\caption{Definition of $f'$} \label{fig:f0'}
\end{figure}

\section{Serialization, compression, and factorization}\label{S:serial}

This section  establishes (in Theorem~\ref{t:dp:byhand}) that each DP-poly-cost 
$\Smi$-computable
function $f\of\gamma_1\to\gamma_0$ has a factorization:
\begin{align*}  
   &  f = \deserialize_{\gamma_0}\circ \widehat{f} \circ \serialize_{\gamma_1}
\end{align*}
where 
\begin{itemize}
  \item $\serialize_{\gamma_1}$  DP-poly-cost maps each $\gamma_1$-value to 
  a particular hereditarily sequential 
  $\sim$-representation (Definition~\ref{d:val:rep}\eqref{i:simrep}) 
  of that value,
  \item $\deserialize_{\gamma_0}$ TD-poly-cost maps 
  each hereditarily sequential representation of a $\gamma_0$-value to 
  the $\gamma_0$-value so represented, and 
  \item $\widehat{f}$ is TD-poly-cost $\Smi$-computable function.
\end{itemize}
Thanks to Theorem~\ref{t:hs}, both 
$\widehat{f}$   and  $\deserialize_{\gamma_0}$ above are
$\RSmi$-computable.   Thus, this factorization 
serves to isolate the completeness question to the 
$\serialize_{\gamma_1}$ functions. 

\subsection{A hereditarily sequential representation of values}

For each $\gamma$, let $\cquote{\gamma}$ be the 
$\String$ version of $\gamma$ where we have fixed 
some sensible $\String$-representation of the syntax of $\Smi$ ground-types.
%
Let $\Vtg = \List_\Vertex$ where:
\begin{align*}
  &\mathsf{datatype} \, \Vertex = \VERT_{()} 
     \synsep \VERT_{\inj_1} \Of \String\times\Nat
     \synsep \VERT_{\inj_2} \Of \String\times\Nat
  \\
  &\Quad{7}
     \synsep \VERT_{(,)} \Of \String\times\Nat\times\Nat
     \synsep \VERT_{\mu} \Of \String\times\Nat
\end{align*}

\begin{definition}\label{d:val:rep} 
Suppose $L=[u_{n-1},\dots,u_1,u_0]$ is a $\Vtg$-value.
\begin{asparaenum}[(a)]
  \item
  	The $L$-item $u_i$ is  said to have $L$-\emph{address} $i$. 
  \item
  	$L$ \emph{represents} a value $v$ iff there is a one-to-one map 
	between $v$'s vertices 	and $L$'s items  such that:
	\begin{itemize}
	  \item 
	    each $\uep$-vertex maps to a $\VERT_{()}$-item;
	  \item 
	  	each type-$(\gamma_1\times\gamma_2)$ vertex $\upair{v_1}{v_2}$
	  	maps to an  
		$\VERT_{(,)} (\cquote{\gamma_1\times\gamma_2},\allowbreak a_1, a_2)$
		where  $a_j$  is the $L$-address of $v_j$'s item, for $j=1,2$,
	  \item 
	    each type-$(\gamma_1+\gamma_2)$ vertex $(\uinj_j\,v_j)$
	  	maps to a $\VERT_{\inj_j} (\cquote{\gamma_1+\gamma_2},
		\allowbreak a_j)$ 
		where $a_j$ is the $L$-address for $v_j$'s item; 
	  \item 
  		each type-$\muP$ vertex $(\uCn{\muP}\,v_0)$
  		maps to a $\VERT_{\mu} (\cquote{\muP\,},a_0)$ where 
  	    $a_0$ is the $L$-address for $v_0$'s item; and
	  \item the order of the items in $L$ corresponds to a topological
	    sort of $v$.
\end{itemize}  
  \item 
  \label{i:simrep}%
  $L$ $\sim$-\emph{represents} a value $v$ iff there 
  is some $v'\sim v$ such that $L$ represents $v'$. 
\end{asparaenum}\end{definition}

\begin{lemma}\label{l:rep:size:bnd}
   For each $\gamma$, there is a quadratic function $q_\gamma(\cdot)$
   such that, for each type-$\gamma$ value $v$, if $L$ represents
   $v$, then $\size(L) \leq q_\gamma(\size(v))$. 
\end{lemma}

\begin{proof}
It is straightforward that for each $\gamma$, there is a finite
set of types $T_\gamma$ such that, for each type-$\gamma$-value $v$, 
the types labeling vertices in $v$ are all from $T_\gamma$. 
Now, fix $\gamma$, a type-$\gamma$-value $v$, and an $L$ that 
represents $v$.  Let $n_v$ be the total number of vertices in $v$, 
and hence,
the number of items in $L$. Thus, if $a$ is an $L$-address occurring
within an $L$-item, then $\size(a) \leq n_v+1$.  Let 
$t_\gamma=\max\set{\size(\cquote{\sigma}) \suchthat \sigma\in T_\gamma}$. 
Then it follows from Definition~\ref{d:val:rep} that 
$2 n_v+t_\gamma +1$ is an upper bound on the size of any $L$-item.  
Hence, $\size(L) \leq n_v\cdot(2 n_v+t_\gamma +2)$.  
By Lemma~\ref{l:size:bnd}, $n_v\leq k_\gamma\cdot(1+\size(v))$. 
Thus, the  quadratic bound on $\size(L)$ follows. 
\end{proof}

\begin{lemma}[Deserialization]\label{l:deser}
  For each $\gamma$, there is a TD-poly cost 
  $\Smi$-function 
  $\deserialize_\gamma\of\Vtg\to\gamma$ such that if
  $(\Gamma\entails e \of\Vtg)\theta \yields L$
  where $L$ represents a type-$\gamma$-value $v$,
  then $(\Gamma\entails \deserialize_\gamma(e)\of\gamma)\theta\yields v$. 
\end{lemma}

\begin{proofsketch}
Fix some default type-$\gamma$-value for $\deserialize_\gamma$ to 
return when $L$ fails to represents a type-$\gamma$-value.
Testing whether a $\Vtg$-value does represent some type-$\gamma$-value $v$,
and if so, constructing this $v$ are fairly standard tasks. 
We leave the details to the reader. 
\end{proofsketch}

\subsection{Compression}

For each $\gamma$, we want $\serialize_\gamma$ to be a $\Smi$-function 
that maps each type-$\gamma$-value $v$ to a $\Vtg$-value that
$\sim$-represents $v$. 
Suppose for the moment $\serialize_\gamma$ has this property
and that $v$ and $v'$ are type-$\gamma$ values
with $v\sim v'$.   Then, as $v\sim v'$, 
we have $\serialize_\gamma(v)\sim\serialize_\gamma(v')$.
By $\Vtg$'s  sequentiality and Definition~\ref{d:val:rep},
it follows that $\serialize_\gamma(v)$ and $\serialize_\gamma(v')$
represent the identical value.
Thus
$\serialize_\gamma$ essentially collapses each bisimilarity-equivalence 
class to a single representative element.
To write $\serialize_\gamma$ we thus need to determine what 
this representative element should be.  One choice, the 
no-sharing/tree $\sim$-equivalent of each value, is ruled out by 
the exponential size blow-ups we  observed in \S\ref{S:intro}.
Instead we choose the $\sim$-equivalent value with \emph{maximal sharing}
and thus minimum size. 
The task of finding this value (it turns out to be unique) is a version 
of the \emph{common subexpression elimination problem} that 
arises in compiler optimization.  
There is a well-known
linear-time algorithm for this problem due to  \citet{DST:comm:sub80}.\sidenote{The linear-time bound is for a RAM model of computation.}
In the terminology of this paper, the algorithm takes a value $v$
(in something like an adjacency-list representation of the dag) and returns 
a bisimilar value $v'$ with $\size(v') = \cs_\gamma(v)$. 
We shall call this $v'$ the \emph{compressed version of $v$}.
By making use of this algorithm (adapted for $\Smi$) we can show:

\begin{lemma}[Serialization]\label{l:compress}
  For each $\gamma$, there is a closed DP-poly-cost $\Smi$-function
  $\serialize_\gamma\of\gamma\to\Vtg$ such that
  if $(\Gamma\entails e\of\gamma)\theta\yields v$, then
  $(\Gamma\entails \serialize(e)\of\Vtg)\theta\yields L$, 
  where the $\Vtg$-list $L$ represents the compressed form of $v$. 
\end{lemma}

\begin{proofsketch}
  We proceed by induction on the structure of $\gamma$. 
  
  \CASE $\gamma=\Unit$.   Then let $\serialize_\gamma (x) =
  [\VERT_{()}]$, which is clearly TD-poly-cost. 
  
  \CASE $\gamma=\gamma_1+\gamma_2$ where 
  the lemma holds for $\gamma_1$ and $\gamma_2$. 
  Then let: 
  \begin{gather*}
  \serialize_\gamma(x) = \\ \Quad1
  \Case x \Of \,\Set{(\inj_j x_j) \Rightarrow 
  \begin{array}{l} \Let \, n_j = \hbox{the length of $\serialize_{\gamma_j}(x_j)$} \\[0.5ex]
  \; \In\;
  \left(\VERT_{\inj_j} (\cquote{\gamma},n_j-1)\right)::\left(\serialize_{\gamma_j}(x_j)\right)
  \end{array}
  }_{j=1,2}
  \end{gather*} 
  Then since both $\serialize_{\gamma_1}$ and 
  $\serialize_{\gamma_2}$ are TD-poly-cost,  it follows that 
  $\serialize_\gamma$ is too. 
  
  \CASE $\gamma=\gamma_1\times\gamma_2$ where 
  the lemma holds for $\gamma_1$ and $\gamma_2$. 
  Then let: 
  \begin{align*}
  \lefteqn{\serialize_\gamma(x) =}\\
   & \Quad1 \Let\,     L_1 = \serialize_{\gamma_1}(\pi_1(x)); \\
   & \Quad{2.6} L_2 = \serialize_{\gamma_2}(\pi_2(x)); \\
   & \Quad{2.6} L = \parbox[t]{8.5cm}{\raggedright the
   $\Vtg$-representation of the compressed version of $x$'s value computed
   by the dag compression algorithm using $L_1$ and $L_2$}
   \\
   &\Quad1\;\;\;\In\;L 
  \end{align*}  
  From $L_1$ and $L_2$ we have complete 
  information about a type-$(\gamma_1\times\gamma_2)$ value $v'$ 
  that is bisimilar to $x$'s value and   
  whose
  two branches are compressed but disjoint. Clearly this is enough
  information to run the dag compression algorithm on $v'$.  
  Hence, we can construct $L$ as required in polynomial time.
  
  \SCASE $\gamma=\Ltree = (\Leaf \synsep \Fork \Of 
  \gamma_0\times\Ltree\times\Ltree$)  
  where the lemma holds for $\gamma_0$. 
  Then $\serialize_\gamma$ is given by:
  \begin{align*}
    & \serialize_{\Ltree} (\Leaf) \\
    & \Quad1 = 
    [\VERT_\mu(\cquote{\,\Ltree},1),\;
     \VERT_{\inj_1}(\cquote{\Unit + \gamma_0\times \Ltree \times \Ltree)},0),\;
     \VERT_{()}]
     \\
    & \serialize_{\Ltree} (\,\Fork(a_0,t_1,t_2)\,) \\
    & \Quad1 = \Let \, L_0 = \serialize_{\gamma_0}(a_0)\\
    & \Quad{3.9} L_1 = \serialize_\Ltree(t_1) \\    
    & \Quad{3.9} L_2 = \serialize_\Ltree(t_1) \\    
    & \Quad{3.9} L = \parbox[t]{8.5cm}{\raggedright the
   $\Vtg$-representation of the compressed version of $\Fork(a_0,t_1,t_2)$'s value computed
   by the dag compression algorithm using $L_0$, $L_1$, and $L_2$}  \\
   &\Quad3 \In\;L 
  \end{align*}
It follows as in the $(\gamma_1\times\gamma_2)$-case that we have 
enough information to compute $L$ as required.  To show that 
$\serialize_{\gamma}$ is DP-poly-cost let us first note that by
Definition~\ref{d:val:rep} and Lemma~\ref{l:rep:size:bnd},
in the computation of $\serialize_{\gamma}(x)$, 
every $L_0$, $L_1$, $L_2$, and $L$ in each step of the recursion is
of size $\leq q_\gamma(\size(x$'s value)).  It follows then that
there is a polynomial function $q(\cdot)$ such that
the cost of each step of the recursion is $\leq q(\size(x$'s value)).
As we are following the dynamic programming evaluation strategy,
there are at most $(1+\size(x$'s value))-many steps of the recursion.  
Thus, there is polynomial in $\size(x$'s value) that bounds the cost of the 
entire recursion.

\CASE $\gamma=\muP$.  This general case is just the prior special case
with more bureaucracy. 
\end{proofsketch}

\begin{corollary}\label{c:cs}
  For each $\gamma$, $\cs_\gamma$ is  DP-poly-cost $\Smi$-computable.
\end{corollary}

\begin{proof} 
  Given $x\of\gamma$, to compute $\cs_\gamma(x$'s value), all one
  needs to do is count the number of data-type constructors represented 
  in $\serialize_\gamma(x)$.  Clearly this is DP-poly-cost task. 
\end{proof}

\subsection{The factorization theorem} 

\begin{theorem}\label{t:dp:byhand}
  Suppose $f$ is a 
  DP-poly-cost $\Smi$-function with $\entails f\of\gamma_1\to\gamma_0$.  
  Then there is a TD-poly-cost $\Smi$-function 
  $\entails \widehat{f}\of\Vtg\to\Vtg$ such that
  \begin{gather*}  
    f = \deserialize_{\gamma_0}\circ \widehat{f} \circ \serialize_{\gamma_1}.
  \end{gather*}
\end{theorem}

\begin{proof} 
We  take $\widehat{f}$ to be a translation of $f$ that computes
over serialized values.  For each application of a constructor
or destructor in $f$ we produce a TD-poly-cost $\Smi$-function
that performs the same task over serialized values.  For example, 
suppose we have a subexpression $(\proj_1\,e)$ with $e\of\gamma_1\times\gamma_2$.  It is simple to construct a 
TD-poly-cost $\Smi$-function $proj_{1,\gamma_1\times\gamma_2}$ 
with $\entails proj_{1,\gamma_1\times\gamma_2}\of\Vtg\to\Vtg$ 
such that: 
\begin{align*}
   \deserialize_{\gamma_1}\circ proj_{1,\gamma_1\times\gamma_2} \circ \serialize_{\gamma_1\times\gamma_2} \equiv 
   \lam{x\of\gamma_1\times\gamma_2}(\proj_1 x).
\end{align*}
So the translation of $(\proj_1\,e)$ would be 
$(proj_{1,\gamma_1\times\gamma_2}\;\hat{e})$ where $\hat{e}$ is the 
translation of $e$.  
Translating $\fold$-expressions is more involved.  
Suppose that $(\fold_{\muP} g \;e)$ is a subexpression of $f$ 
and that $\hat{g}$ and $\hat{e}$ are the respective 
translations of $g$ and $e$.   
So the translation of 
$(\fold_{\muP} g \;e)$ will be of the form
$h(\fold_{\Vtg}\, \hat{g}^* \;\hat{e})$ where $\hat{g}^*$ carries
out the steps of a dynamic programming evaluation of 
$(\fold_{\muP} g \;e)$ (on serialized values via $\hat{g}$) using 
the serialized (and topologically sorted) version of $e$'s value
provided by $\hat{e}$'s value and where $h$ extracts the appropriate final 
value from $(\fold_{\Vtg}\, \hat{g}^* \;\hat{e})$'s result. 
Given that $\hat{g}$ and $\hat{e}$ are TD-poly-cost, the
indicated translation of $(\fold_{\muP} g \;e)$ should be also. 
\end{proof}

\section{Completing $\RSmi$}\label{S:completing}

First let us note:

\begin{lemma}  \label{l:2/3} \
\begin{asparaenum}[(a)]
  \item \label{i:deser:RSmii}
  For each $\gamma$, $\deserialize_\gamma$ is $\RSmi$-computable. 
  \item \label{i:hatf}
    Suppose $f$ is a 
  DP-poly-cost $\Smi$-function with $\entails    
  f\of\gamma_1\to\gamma_0$.  
  Then $\widehat{f}$ as in Theorem~\ref{t:dp:byhand} is $\RSmi$-computable.
\end{asparaenum}  
\end{lemma}

\begin{proof}
  Theorem~\ref{t:hs} and Lemma~\ref{l:deser} together yield part (a).
  Theorems~\ref{t:hs} and~\ref{t:dp:byhand} give us part (b). 
\end{proof}

Thus by Theorem~\ref{t:dp:byhand}'s factorization, 
if the $\serialize_\gamma$ functions were $\RSmi$-computable,
then $\RSmi$ would be complete for the DP-poly-cost $\Smi$-computable 
functions.  But when we try to define the $\serialize_\gamma$ functions
within $\RSmi$ we face essentially the same problem we encountered
with the $\height$ function in \S\ref{S:intro}:  we seem to be 
blocked from combining the safe results of multiple branches of a 
recursion unless we 
an \emph{a priori}, normal upper bound on the sizes of these safe results. 
To get around this problem we extend $\RSmi$ to a new formalism, 
$\RSmii$. 

\begin{figure}[t]
\begin{align*}
 & E \;\is\; \dots
        \; \synsep (\CS{\delta}\, E) && \hbox{\sc Syntactic extensions}
\\&			
 \rulelabel{$\CS\delta$-I} 
  \irule{\Gamma\entails e\of \delta}%
        {\Gamma\entails (\CS\delta\, e)\of\Nat}
         && \hbox{\sc Typing extentions}
\\
&
  \rulelabel{Compressed Size}
  \irule{e\theta\yields v }{
   (\CS\gamma \,e)\theta\yields \cs(v) }
         && \hbox{\sc Evaluation semantics extensions}   
\end{align*} 
\caption{Extensions for $\RSmii$ where $\delta=\muP$}\label{f:Rsmii:rev}
\end{figure}

\begin{definition}
  $\RSmii$ is the extension of $\RSmi$ that results 
  from adding the $\cs_\delta$'s as new initial functions, 
  i.e., on the same level as the $\Cn\delta$'s and $\Dn\delta$'s.  
  The details of the syntactic, typing, and semantic extensions
  are given in Figure~\ref{f:Rsmii:rev}.
\end{definition}

\begin{theorem}[Polynomial-time soundness for $\RSmii$]\label{t:rsmii:sound}
   Suppose $\entails_{\RSmii} f\of\gamma_1\to\gamma_0$ where 
   $\gamma_1$ is normal.  Then $f$ is DP-poly-cost. 
\end{theorem}

\begin{proof}
The theorem follows by Theorem~\ref{t:rsmi:sound:1}, the poly-time 
soundness of $\RSmi$, and Corollary~\ref{c:cs}, 
the fact that each $\cs_\gamma$ is 
DP-poly-cost $\Smi$-computable. 
\end{proof}

Here is the key theorem.

\begin{theorem}\label{t:ser:RSmii}
  For each $\gamma$, $\serialize_\gamma$ is $\RSmii$-computable. 
\end{theorem}

For this theorem we need a slight reworking of Theorem~\ref{t:hs}
for use within ramified structural recursions.

\begin{lemma} \label{l:hs:mod}   
  Suppose $\entails_{\Smi}f\of\gamma_1\to\gamma_0$ where 
  $\gamma_1$ is normal and hereditarily sequential and where $f$ 
  is TD-poly-cost. 
  Then there is a $\RSmi$-function $f^\star$ with 
  $\entails_{\RSmi}f^\star\of\Nat\times\safe{\gamma_1}\to\safe{\gamma_0}$
  such that, for all $\theta_0\of(x_1\of\gamma_1)$ and all
  $m\geq \size(\theta_0(x_1))$, if $f(x_1)\theta_0 \yields v$, 
  then  $\toNorm(f^*(\numeral{m},\toSafe(x_1)))\theta_0\yields v$. 
\end{lemma}

\begin{proof} 
Let $f^* = (\lam{(n,x_1)}e)$ where $e$ is the body of $f'$
from Figure~\ref{fig:f0'} with the following two modifications:
\begin{inparaenum}[(i)]
  \item the type of $as_0$ is changed to $\safe\Alist$ and
  \item the definition of $b_0$ is changed to 
       an expression computing $\numeral{3\cdot q_0(n\hbox{'s value})}$.   
\end{inparaenum}
Since the only use of $as_0$ is in defining $s_0$, a safe value,
this first change causes no typing problems.   
Now fix $\theta_0$ and an  $m\geq \size(x_1$'s value).
Suppose $f'(x_1)\theta_0 \yields v$.
Then it follows from the proof of Theorem~\ref{t:hs} that 
$\toNorm(f^*(\numeral{m},\toSafe(x_1)))\theta_0\yields v$ as required.
\end{proof}

\begin{proof}[of Theorem~\ref{t:ser:RSmii}]
  We proceed by induction on the structure of $\gamma$. 
  The $\gamma=\Unit$, $\gamma=\gamma_1+\gamma_2$, and 
  $\gamma=\gamma_1\times\gamma_2$ cases are exactly as in the
  proof of Lemma~\ref{l:compress}.  As in 
  Lemma~\ref{l:compress}'s proof, we consider the 
  special case of $\gamma=\Ltree = (\Leaf \synsep \Fork \Of 
  \gamma_0\times\Ltree\times\Ltree$)  
  where the lemma holds for $\gamma_0$. 
  Let $f_0$ be a $\Smi$-function with 
  $\entails_{\Smi} f_0\of \Vtg\times\Vtg\times\Vtg\to\Vtg$
  such that $f_0(L_0,L_1,L_2)$ computes $L$ as in $\Fork$-case
  of the structural recursion for $\serialize_{\Ltree}$ given
  in the proof of Lemma~\ref{l:compress}.
  Moreover, we can take $f_0$ to be TD-poly-cost. 
  Hence, for this $f_0$, there is, by Lemma~\ref{l:hs:mod},
  an $\RSmi$-function $f_0^*$ as in the lemma. 
  (So, $\entails f_0^*\of \Nat \times 
  \safe{(\Vtg\times\Vtg\times\Vtg)} \to \safe\Vtg$.)
  It follows by Lemma~\ref{l:compress} that there is a polynomial 
  $q_\Ltree$ such that for all $\Ltree$-values $v$, 
  $\size(\serialize_\Ltree(v)) \leq q_\Ltree(\cs(v))$. 
  We can now write an $\RSmii$-version of $\serialize_\Ltree$
  as follows.
  \vspace*{-2ex}
  \begin{leftbar}
  \begin{align*}
    & \serialize_\Ltree(x) = \\
    & \Quad1 \Let\, (n\of\Nat)=q_\Ltree(\CS{\Ltree}(x))\\
    & \Quad{2.6}  \step \of (\Unit+\gamma_0\times\safe\Vtg\times\safe\Vtg) \to\safe\Vtg \\
    & \Quad{2.6}  \step (\inj_1()) = 
    \toSafe\left(\left[\VERT_\mu(\cquote{\,\Ltree},1),\;
     \VERT_{\inj_1}(\cquote{\Unit + \gamma_0\times \Ltree \times \Ltree)},0),\;
     \VERT_{()}\right]\right)
     \\
    & \Quad{2.6} \step (\,\inj_2 (a_0,L_1,L_2)\,) 
    = f_0^*\left(n,\, (\toSafe(\serialize_{\gamma_0}(a_0)),L_1,L_2)\,\right)
    \\
   &\Quad{1.2} \In\;(\fold_{\Vtg} \,\step \,x )
  \end{align*}\end{leftbar}
  Note that in the $\step (\,\inj_2 (a_0,L_1,L_2)\,)$ case, 
  we have that 
  $\size($the value of $(\toSafe(\serialize_{\gamma_0}(a_0)),L_1,L_2))
  \leq n$'s value.
  It follows from a straightforward induction that the above definition 
  is as required. 
  As in the proof of Lemma~\ref{l:compress}, the general 
  data-type case is just the  special case with more bureaucracy. 
\end{proof}

\begin{theorem}[$\RSmii$ completeness] 
  \label{t:rsii:comp}
  Each DP-poly-cost $\Smi$-function is $\RSmii$-computable.
\end{theorem}

\begin{proof}
  The theorem follows immediately from 
  Lemma~\ref{l:2/3} and
  Theorems~\ref{t:dp:byhand} and~\ref{t:ser:RSmii}.
\end{proof}

\section{Open Questions} 

The first and most obvious open question is 
what is 
the resolution of 
our conjecture on the incompleteness of $\RSmi$ and the 
related formalisms of \citep{gen:rem:2010,Avanzini2018OnSM}. 
We have tried various approaches and had no luck with any of them. 

The
 notion of ``inductively defined data type'' used in this paper 
omits several standard structures, e.g., rose trees \citep{BirdGibbons}.
With very strong notions of inductively defined data (e.g., 
$\lfp{t}(\Unit+(\Nat\to t))$) not only will the techniques of this 
paper break, but it is open what ``feasibly computable'' should 
mean in such a context. 

On a different tack, we seem to have two distinct notions of 
feasibility for computation over inductively defined data: (1)
the no-sharing notion of POLA \citep{Burrell:2009}
and (2) the maximum-sharing notion explored in this paper.  
It would be nice to better understand the strengths and weakness 
of the two approaches and the trade-offs between them.    
We suspect that neither of these two approaches matches 
the intuitions of a practiced functional programmer as to what
feasibility computation on inductively defined data should be. 
It is thus worth asking if there are other interesting approaches
to be explored. 


\bibliographystyle{plainnat}
\bibliography{ops}


\bigskip
\appendix 

\section{Soundness results for $\RSmi$}\label{S:soundness}

Theorem~\ref{t:rsmi:sound:1}, polynomial-time soundness for $\RSmi$,
follows from more general results in \citep{DR:23}.  As
of this writing that paper is not yet public.
This appendix fills the gap by sketching  the basic soundness results for 
$\RSmi$, including a normal/safe noninterference theorem in \S\ref{S:noninterf}.

\subsection{Normal and safe spans and residual sizes}\label{S:res:size}

\emph{Conventions on digraphs:} 
$\emptydag=(\emptyset,\emptyset)$, the empty graph.
For  $G=(V,E)$ and $G'=(V',E')$, 
$G\cup G'= (V\cup V', E\cup E')$,
$G\cap G'= (V\cap V', E\cap E')$,  and
$G\setminus G' =$
the subgraph of $G$ induced by $(V\setminus V')$. 
Also, $(a,b)_\ell$ denotes an edge, with label $\ell$, from vertex $a$ to vertex $b$.

\begin{definition} 
\ \label{d:spans}
\begin{asparaenum}[(a)] 
  \item
  For each value term graph $v$, let $r_v$ denote the root vertex of $v$. 
  \item \label{i:sp:norm}
  For each  $\gamma$ and each $v$, a type-$\gamma$ value-term graph, define the \emph{$\gamma$-normal span of $v$} (written:  
  $\sharp_\gamma (v)$) to be the 
  subgraph of $v$ given by:
      \begin{align*}
      \sharp_\gamma (v) &= 
      \begin{cases}
    	v, & \hbox{if $\gamma$ is normal;}\\
    	\emptydag,  & \hbox{if $\gamma$ is safe;}\\
	    (\set{r_v}\cup V,E\cup E'), 
		     & \hbox{if }
		     \parbox[t]{7cm}{\raggedright $\gamma=\gamma_1+\gamma_2$ 
		        is mixed and $v=\uinj_i(v_0)$, \\
		        where  $(V,E)=
		           \sharp_{\gamma_i}(v_0)$ and 
		           $E'=\set{(r_v,r_{v_0}) 
		           \suchthat \sharp_{\gamma_i}(v_0)\not=\emptydag}$}\\   
	    (\set{r_v}\cup V,E\cup E'), 
		     & \hbox{if }
		     \parbox[t]{7.35cm}{\raggedright $\gamma=\gamma_1\times\gamma_2$ 
		        is mixed and $v=\upair{v_1}{v_2}$,\\
		        where  $(V,E)=
		           \sharp_{\gamma_1}(v_1)\cup
		           \sharp_{\gamma_2}(v_2)$ and \\
		           $E'=\set{(r_v,r_{v_i})_{\pi_i} 
		           \suchthat i=1,2 \;\;\&\;\;\sharp_{\gamma_i}(v_i)\not=\emptydag}$.}
      \end{cases}    
    \end{align*}
  \item \label{i:sp:safe}
  Define  the \emph{$\gamma$-safe span of $v$}
  (written: $\flat_\gamma(v)$) analogously. 
\end{asparaenum}
\end{definition}

%

Note that, for each value-term graph $v$ for a type-$\gamma$ value:
\begin{align*}
    &\sharp_{\norm{\gamma}}(v) = \flat_{\safe{\gamma}}(v) = v.
    &&\sharp_{\safe{\gamma}}(v) = \flat_{\norm{\gamma}}(v) = \emptydag. 
\end{align*}

\begin{definition} \  \label{d:resSize}
\begin{asparaenum}[(a)] 
 \item  Suppose $X\subseteq \dom(\Gamma)$ and $\theta\of\Gamma$.
 Let
   $SS(X)\theta = 
   \bigcup_{x\in X}(\flat_{\Gamma(x)}(\theta(x)))$.
 We call $SS(X)\theta$ the \emph{safe span of $X$ with respect to $\theta$}. 

  \item 
  Suppose $v$ is a type-$\gamma$ value
  The \emph{$\gamma$-nonnormal span of $v$}
  (written: $\nn\gamma(v)$) is given by:  
  $\nn\gamma(v) = v\setminus\sharp_\gamma(v)$.

  \item \label{i:resSize}
The \emph{residual size} of an $\RSmi$ judgment
$\Gamma\entails e\of\gamma$ is the function
$\residual{\Gamma\entails e\of\gamma}\of \set{\theta \suchthat \theta\of\Gamma}\to\nat$ such that:\marginnote[6ex]{$\dotcup$ denotes disjoint union.}
\begin{gather*}
  \residual{\Gamma\entails e\of\gamma}\theta 
  \;\defeq\; \size(\sharp_\gamma(v)\dotcup 
                   (\nn\gamma(v) \setminus SS(\fv(e))\theta)),
   \hbox{ where $e\theta\yields v$.}
\end{gather*}
We usually write $\residual{e}$ for 
$\residual{\Gamma\entails e\of\gamma}$
when $\Gamma\entails e\of\gamma$ is understood.
\end{asparaenum}
\end{definition}

Suppose $e\theta\yields v$.  
Then, 
intuitively, $\residual{e}\theta$ is a sum of the size of 
the normal part of $v$ plus the size of the 
``potentially normal'' part of $v$.  By ``potentially normal'' we 
mean that $\nn\gamma(v) \setminus SS(\fv(e))\theta$ is the 
non-normal part of $v$ that is not claimed as safe by any of 
$e$'s free variables, and hence, it might be eventually 
reclassified as normal by means of $\toNorm$.   
\emph{Example 1:}  Here is a typical calculation.  Suppose that
$x\of\gamma_0,y\of\gamma_1\entails e\of \gamma_2$ where $\gamma_0$
is normal and $\gamma_1$ and $\gamma_2$ are safe and that $e\theta\yields v$.  Then
$\residual{e}\theta = \size(\sharp_{\gamma_2}(v) \dotcup \nn{\gamma_2}(v)
\setminus SS(\set{x,y})\theta) = 
\size(\sharp_{\gamma_2}(v) \dotcup (v\setminus \sharp_{\gamma_2}(v))
\setminus (\flat_{\gamma_0}(\theta(x)) \cup \flat_{\gamma_1}(\theta(y))))$
which, as $\gamma_0$ is normal and $\gamma_1$ and $\gamma_2$ are safe, is 
$=  \size(\emptydag \dotcup (v\setminus \emptydag)
\setminus (\emptydag \cup \theta(y)) = \size(v\setminus \theta(y))$. 
\emph{Example 2:} As shown in the proof of Theorem~\ref{t:sizebnd}:
$\residual{(\Cs\delta e_0)}\theta = 1+\residual{e_0}\theta$.
That ``$1+$'' is from the newly created $\uCn\delta$-vertex 
which is in the safe part of the value, but, as it is 
freshly created,  not in the safe span of $e$'s free variables.
For a key part of the motivation for using this odd-looking notion 
of size, see Scholium~\ref{s:res:pair} below. 

\begin{lemma}  \label{l:res}
    Suppose $\Gamma(x)=\gamma$ and $\theta\of\Gamma$. Then:
  $\residual{x}\theta = \size(\sharp_\gamma(\theta(x)))$. 
%
So, $\residual{x}\theta=\size(\theta(x))$ when $\gamma$ is normal
and $\residual{x}\theta=0$ when $\gamma$ is safe.
\end{lemma}
\begin{proof}  Note: 
$\residual{x}\theta 
  = \size(\sharp_\gamma(\theta(x)))$ +
                   $\size(\nn\gamma(\theta(x)) \setminus SS(\fv(x))\theta))$
and  
$\nn\gamma(\theta(x)) \setminus SS(\fv(x))\theta))
=  (\theta(x)\setminus\sharp_\gamma(\theta(x)))\setminus
\flat_\gamma(\theta(x)) = \emptydag$.
\end{proof}

\begin{definition}
Suppose $\set{x_1,\dots,x_k}\subseteq \dom(\Gamma)$,
 $q$ is a polynomial  over indeterminates
  $\residual{x_1},\dots,\residual{x_k}$, and $\theta\of\Gamma$.
Then $q\theta\in\nat$ is given by:
\begin{align*}
  &q  \theta \;\defeq \;
  q\left[\strut\residual{x_1}\gets \residual{x_1}\theta,\dots,
   \residual{x_k}\gets \residual{x_k}\theta\right]. 
\end{align*}
\end{definition}

\subsection{A polynomial size bound}

\begin{theorem}\label{t:sizebnd}
  Suppose  $\Gamma\entails_{\RSmi} e\of\gamma$   with $\fv(e)=\set{x_1,\dots,x_k}$.
  Then there is a polynomial $q$ over indeterminates 
  $\residual{x_1},\dots,\residual{x_k}$ such that, for all 
  $\theta\of\Gamma$, we have
     $\residual{e}\theta
     \;\leq\; q \theta$. 
\end{theorem}

For the proof of Theorem~\ref{t:sizebnd} we need two bits of set/dag algebra.

\begin{lemma}  \label{l:setids}
Suppose $A$, $B$, $C$, and $D$ are sets.\sidenote{Thanks to the digraph conventions  at the beginning of \S\ref{S:res:size}, the lemma extends to digraphs.}
\begin{asparaenum}[(a)]
  \item \label{i:pair}
  $(A\cup C)\setminus (B\cup D) \subseteq (A\setminus B)\cup (C\setminus D)$.
  \item \label{i:app}
  $A\setminus ((B\setminus C)\cup D) 
  = (A \setminus (B\cup D)) \cup ((A\cap C)\setminus D)$.
\end{asparaenum}
\end{lemma}
\begin{proof}
Part (a) is straightforward.  For part (b),
let $\compl{X} = ((A\cup B\cup C\cup D)\setminus X)$.  Then:
\begin{align*}
A\setminus ((B\setminus C)\cup D) 
  &= A \cap (\compl{B}\cup C)\cap\compl{D} \\
  &= (A \cap \compl{B}\cap\compl{D})\cup (A\cap C\cap \compl{D}) \\
  &= (A\setminus (B\cup D)) \cup ((A\cap C)\setminus D).
\end{align*}  
\end{proof}

\begin{scholium}[Residual size and sharing]\label{s:res:pair}
 The size bounds used by \citet{BellantoniCook} 
were of the form $p+\max(|y_1|,\dots,|y_n|)$
where $p$ was a polynomial over the sizes of normal variables 
and $y_1,\dots,y_n$ were  safe variables.  These poly-max bounds
do not combine well when trying to size-bound a pairing
$(e_1,e_2)$ in terms
of size-bounds for $e_1$ and $e_2$.  In contrast,
the $e=(e_1,e_2)$ case of the proof of Theorem~\ref{t:sizebnd}
below shows that if $q_1$ and $q_2$ are polynomial bounds
for $\residual{e_1}$ and $\residual{e_2}$, respectively, then
$q_1+q_2$ serves as a bound on $\residual{(e_1,e_2)}$.  
Thus while we bounding a more complex notion of size than that used by
Bellantoni and Cook, our bounds are simpler to work with
than poly-max bounds. 
\end{scholium}

\begin{proofsketch}[of Theorem~\ref{t:sizebnd}]
We proceed by strong induction on the derivation of $\Gamma\entails e\of\gamma$.  Suppose $\theta\of\Gamma$.

\CASES $e=()$ and $e=\safe{()}$.  
It clearly suffices to take $q=0$.

\CASE $e=x$, a variable.
It clearly suffices to take  $q=\residual{x}$.

\CASE $e=\toSafe(e_0)$, 
where $\Gamma\entails e_0\of\gamma_0$
and $\Gamma\entails \toSafe(e_0)\of\safe{\gamma_0}$.
By the IH for $e_0$, there is a polynomial bound $q_0$ for $\residual{e_0}$.
Suppose $e_0\theta\yields v$, hence, $e\theta\yields v$ also. 
Then 
\begin{align*}
 \residual{e}\theta 
   & = \size(\sharp_{\safe{\gamma_0}}(v) \dotcup (
       \nn{\safe{\gamma_0}}(v)\setminus SS(\fv(e))\theta))
   && \hbox{(by Definition~\ref{d:resSize}\eqref{i:resSize})}    
\\
  &= \size(\emptydag \dotcup (v\setminus SS(\fv(e_0))\theta))
   && \hbox{(since $\safe{\gamma_0}$ is safe and $\fv(e)=\fv(e_0)$)}    
  \\
  &\leq \size(\sharp_{\gamma_0}(v) \dotcup ((v\setminus \sharp_{\gamma_0}(v))
\setminus SS(\fv(e_0))\theta)) 
   && \hbox{(since $v=\sharp_{\gamma_0}(v)\dotcup\nn{\gamma_0}(v)$)}    
\\
  &\leq \residual{e_0}\theta
   && \hbox{(by Definition~\ref{d:resSize}\eqref{i:resSize})}    
\\
  &\leq q_0\theta
   && \hbox{(by the IH on $e_0$).}    
\end{align*} 
It clearly suffices to take   $q=q_0$.

\CASE $e=\toNorm(e_0)$, 
where $\Gamma\entails e_0\of\gamma_0$
and $\Gamma\entails \toNorm(e_0)\of\norm{\gamma_0}$.
By the IH for $e_0$, there is a polynomial bound $q_0$ for $\residual{e_0}$.
Suppose $e_0\theta\yields v$, hence, $e\theta\yields v$ also. 
By the side-condition on \emph{$\toNorm$-I}, $\fv(e_0)$ consists of 
variables of normal types,  hence, $SS(\fv(e_0))\theta=\emptydag$. 
Thus,
\begin{gather*}
 \residual{e}\theta = \size(\sharp_{\norm{\gamma_0}}(v))
  = \size(v)
   = \size(\sharp_{\gamma_0}(v)\dotcup \nn{\gamma_0}(v))
   \\
   \quad = \size(\sharp_{\gamma_0}(v)\dotcup (\nn{\gamma_0}(v)\setminus SS(\fv(e_0))))
   = \residual{e_0}\theta 
    \leq q_0\theta. 
\end{gather*}
It clearly suffices to take   $q=q_0$.

\CASES $e=(\proj_j e_0)$, $e=(\inj_j e_0)$, $e=(\Dn\delta e_0)$, and 
$(\Ds\delta e_0)$. 
By the IH for $e_0$, there is a polynomial bound $q_0$ for $\residual{e_0}$. 
In each of these cases, it clearly suffices to take  $q=q_0$. 

\CASE $e=(\Cn\delta e_0)$, where $\delta=\muP$ is normal. 
Suppose $e_0\theta \yields v_0$.  Then $e\theta \yields (\uCn\delta v_0)$
and, since $\delta$ is normal,  
$\residual{e}\theta = \size(\sharp_\delta(\uCn\delta v_0)) = 
\size(\,(\uCn\delta v_0)\,) = 1+\size(v_0) = 1+\size(\sharp_{P\delta}(v_0))
= 1+\residual{e_0}\theta$.  
By the IH for $e_0$, there is a polynomial bound $q_0$ for $\residual{e_0}$. 
It clearly suffices to take  $q=1+q_0$.

\CASE $e=(\Cs\delta e_0)$, where $\delta$ is normal and 
$\Gamma\entails e_0 \of \safe{(P \delta)}$. 
Suppose $e_0\theta \yields v_0$.  Then $e\theta \yields (\uCn\delta v_0)$.
Since $\safe\delta$ is safe,  $\sharp_{\safe\delta}(\,(\uCn\delta v_0)\,)=\emptydag$. Hence,
$\residual{e}\theta = \size((\uCn\delta v_0)\setminus SS(\fv(e)) )$
which, since the $\uCn\delta$-vertex created by $e$ is fresh, is 
$=1 + \size( v_0\setminus SS(\fv(e)) )$ which, since $\safe{(P\delta)}$ is 
safe, is $=1+\residual{e_0}\theta$.  
By the IH for $e_0$, there is a polynomial bound $q_0$ for $\residual{e_0}$. 
It clearly suffices to take  $q=1+q_0$.

\CASE $e=(e_1,e_2)$. 
\emph{Claim:} 
$\residual{(e_1,e_2)}\theta \leq \residual{e_1}\theta + \residual{e_2}\theta$. 
\emph{Proof:} Suppose 
$(e_1,e_2)\theta \yields \upair{v_1}{v_2}$.  Then 
$\residual{(e_1,e_2)}\theta$
\begin{align*}
 &  = \size(\sharp_{\gamma_1\times\gamma_2}(\upair{v_1}{v_2}))
    \Quad{10}
    && \hbox{(by Definition~\ref{d:resSize}\eqref{i:resSize})}
 \\ & \lefteqn{\Quad3 +\size(\nn{\gamma_1\times\gamma_2}(\upair{v_1}{v_2}) \setminus SS(\fv((e_1,e_2)))\theta)}
 \\
 &  \leq \size(\sharp_{\gamma_1}(v_1))+\size(\sharp_{\gamma_2}(v_2))
  && \hbox{(since $\size(\cdot)$ does not count $\underline{(,)}$-vertices)}
 \\
 & \lefteqn{\Quad3
 + \size(\nn{\gamma_1}(v_1)\cup\nn{\gamma_2}(v_2) \setminus 
 SS(\fv(e_1)\cup SS(\fv(e_2))\theta))}
 \\
  &  \leq \size(\sharp_{\gamma_1}(v_1))+\size(\sharp_{\gamma_2}(v_2))
  && \hbox{(by Lemma~\ref{l:setids}\eqref{i:pair})}
  \\& \lefteqn{\Quad3 + \size(\nn{\gamma_1}(v_1)\setminus SS(\fv(e_1))\theta) 
  + \size(\nn{\gamma_2}(v_2)\setminus SS(\fv(e_2))\theta) }
   \\
     & \lefteqn{= \left(\size(\sharp_{\gamma_1}(v_1))
  + \size(\nn{\gamma_1}(v_1)\setminus SS(\fv(e_1))\theta) \right)}
  \\ &\lefteqn{\Quad3 +\left(\size(\sharp_{\gamma_2}(v_2))  
  + \size(\nn{\gamma_2}(v_2)\setminus SS(\fv(e_2))\theta) \right)}
  \\
  &  = \residual{e_1}\theta + \residual{e_2}\theta
    && \hbox{(by Definition~\ref{d:resSize}\eqref{i:resSize}).}  
\end{align*}
For $i=1,2$, 
by the IH for $e_i$, there is a polynomial bound $q_i$ for $\residual{e_i}$.
Clearly, it suffices to take $q=q_1+q_2$.

\CASE $e=((\lam{z_1}e_0)\,e_1)$, where $\Gamma,z_1\of\gamma_1\entails e_0\of\gamma_0$ and $\Gamma\entails e_1\of\gamma_1$.
Suppose: 
\begin{inparaenum}[\em(i)]
  \item $\theta\of\Gamma$, 
  \item $e_1\theta\yields v_1$, 
  \item $\theta_0=\theta[z_1\mapsto v_1]$, and
  \item $e_0\theta_0\yields v$.
\end{inparaenum}  
Thus, $e\theta\yields v$. 
For $i=1,2$,
by the IH for $e_i$, there is a polynomial bound $q_i$ for $\residual{e_i}$. 
Without loss of generality, we assume $z_1\notin\fv(e)$.

\emph{Subcase 1:} $\gamma_1$ is normal.   
Then $\flat_{\gamma_1}(z_1)\theta_0=\emptydag$, and so,
$SS(\fv(e_0))\theta_0 = 
      SS(\fv(e_0)\setminus\set{z_1})\theta_0 \subseteq
      SS(\fv(e))\theta$.
Hence:
\begin{align*}
  \residual{e}\theta &= 
        \size(\sharp_{\gamma_0}(v) + (\nn{\gamma_0}(v)\setminus
                                      SS(\fv(e))\theta))
     &&\hbox{(by Definition~\ref{d:resSize}\eqref{i:resSize})}                                      
     \\
     &\leq \size(\sharp_{\gamma_0}(v) + (\nn{\gamma_0}(v)\setminus
                                      SS(\fv(e_0))\theta_0))
     &&\hbox{(as $SS(\fv(e_0))\theta_0\subseteq SS(\fv(e))\theta$)}                                      
     \\
     &= \residual{e_0}\theta_0
     &&\hbox{(by Definition~\ref{d:resSize}\eqref{i:resSize})}                                      
     \\
     & \leq q_0\theta_0
     &&\hbox{(by the IH on $e_0$)}                                      
     \\
     & \leq q_0[\residual{z_1}\gets q_1]\theta.
\end{align*}
The last inequality follows since 
$z_1\theta_0$ and $e_1\theta$ have the same type and the same value
and since $q_0$ is monotone and
$\residual{e_1}\theta \leq q_1\theta$.
Thus, it clearly suffices to take $q= q_0[\residual{z_1}\gets q_1]$.

\medskip
For the next two subcases and for the $\fold$-case we need the 
following decomposition lemma.

\begin{lemma} \label{l:ugh}
$\sharp_{\gamma_0}(v) \dotcup (\nn{\gamma_0}(v)\setminus SS(\fv(e))\theta) \;\subseteq$
\begin{align*}
& \left(\sharp_{\gamma_0}(v) \dotcup (\nn{\gamma_0}(v)\setminus SS(\fv(e_0))
\theta_0)\right)
\;\cup \;
\left(\sharp_{\gamma_1}(v_1) \dotcup (\nn{\gamma_1}(v_1)\setminus SS(\fv(e_1))\theta)\right). 
\end{align*}
\end{lemma}
\begin{proof} 
Note that:
\begin{align*}
   \lefteqn{\nn{\gamma_0}(v)\setminus SS(\fv(e))\theta}
   \\
   & \quad= \nn{\gamma_0}(v)\setminus \big(SS((\fv(e_0)\setminus\set{z_1})\cup\fv(e_1))\theta\big)
   \\
   & \quad = 
   \left(\nn{\gamma_0}(v)\setminus \big(SS(\fv(e_0)\cup\fv(e_1))\theta_0\right) 
   \Quad4
   \hbox{(by Lemma~\ref{l:setids}\eqref{i:app})}
   \\
   & \Quad3 \bigcup
   \left(\big(\nn{\gamma_0}(v)\cap SS(\set{z_1})\theta_0\big)\setminus SS(\fv(e_1))\theta_0\right)
   \\
   & \quad \subseteq 
   \left(\nn{\gamma_0}(v)\setminus \big(SS(\fv(e_0))\theta_0\right) \\
   &\Quad3
    \bigcup
   \left(\sharp_{\gamma_1}(v_1) \dotcup (\nn{\gamma_1}(v_1)\setminus SS(\fv(e_1))\theta)\right).   
\end{align*}
\end{proof}

\emph{Subcase 2:} $\gamma_1$ is safe.  Note that:
\begin{align*}
   \lefteqn{\size\big(\sharp_{\gamma_0}(v)\dotcup(\nn{\gamma_0}(v)\setminus SS(\fv(e_0))\theta_0)\big) }
   \\
   &\quad = \residual{e_0}\theta_0
   && \hbox{(by Definition~\ref{d:resSize}\eqref{i:resSize})}
   \\
   & \quad \leq q_0\theta_0
   && \hbox{(by the IH for $e_0$)}
   \\
   & \quad \leq q_0[\residual{z_1}\gets0]\theta
   && \hbox{(as $\gamma_1$ is safe, $\residual{z_1}\theta_0=0$).}  \qquad(\star)
\end{align*}
\begin{align*}
   \lefteqn{\size\big(\sharp_{\gamma_1}(v_1)\dotcup (\nn{\gamma_1}(v_1)\setminus SS(\fv(e_1))\theta)\big)}
   \\
   &\quad \leq \residual{e_1}\theta
   &&\hbox{(by Definition~\ref{d:resSize}\eqref{i:resSize})}
   \\
   & \quad \leq q_1\theta
   &&\hbox{(by the IH for $e_1$).}
\end{align*}
Thus by Lemma~\ref{l:ugh}, it clearly suffices to take $q= q_0[\residual{z_1}\gets 0]+q_1$.

\emph{Subcase 3:} $\gamma_1$ is mixed.  
Take  $q=q_0[\residual{z_1}\gets q_1]+q_1$.
The proof that this choice suffices is a 
combination of the arguments for subcases 1 and 2.
The essential change is that in the line marked with a $(\star)$
in subcase 2 is changed to:
\begin{align*}
  & \quad \leq q_0[\residual{z_1}\gets q_1]\theta
  && \hbox{(by the argument from subcase 1).}
\end{align*}

\CASE $e=(\Case e_0 \Of\; (\inj_1 x_1) \Rightarrow e_1;\;
                        (\inj_2 x_2) \Rightarrow e_2)$.
By the IH on $e_0$, $e_1$, and $e_2$, there are polynomial bounds
$q_0$, $q_1$, and $q_2$ on   
$\residual{e_0}$, $\residual{e_1}$, and $\residual{e_2}$, respectively.
Take $q=q_1[\residual{x_1}\gets q_0] + q_2[\residual{x_2}\gets q_0]$.                      
The argument that this bound works  is a straightforward modification of the 
argument for the previous  case.                         

\medskip

Recall that $\List_\gamma = \lfp{t}(\Unit + \gamma\times t)$.
Let $\gamma_*  =\Unit+\gamma_1\times\safe{\gamma_0}$

\SCASE $e=(\fold_{\List_{\gamma_1}} (\lam{z}e_0) \; e_1)$, where 
$\Gamma\entails e_1\of\gamma_1$ and 
$\Gamma,z\of\gamma_* \entails e_0\of \safe{\gamma_0}$.
For $i=1,2$,
by the IH for $e_i$, there is a polynomial bound $q_i$ for $\residual{e_i}$.

Suppose $e_1\theta\yields v_1$ and
${a}_1,\dots,{a}_{k}$ are  (from front to back) the $\gamma_1$-values 
making up the $\List_{\gamma_1}$-value  $v_1$. 
Hence, $k+1$ and each of $\size(a_1),\dots,\size(a_k)$ is $\leq\size(v_1)$ which, since $\gamma_1$ is normal, is 
$\leq \residual{e_1}\theta$ which in turn, by the IH on $e_1$, is 
$\leq q_1\theta$.

For $i=1,\dots,k+1$, let $u_i$ be the $\safe{\gamma_0}$ value such that
$\hat{e}_i\theta\yields u_i$ where\footnote{\label{fn:abuse}%
  The value $a_{i+1}$ in the equation for $\hat{e}_i$ 
  is an abuse of notation which we use to avoid yet more
  notational clutter.}
\begin{gather*}
\hat{e}_{k+1} = ((\lam{z}e_0)\;(\inj_1 {()})) \quad\hbox{and}\quad
   \hat{e}_i = ((\lam{z}e_0)\;(\inj_2 ({a}_{i},\hat{e}_{i+1})))
   \quad (1\leq i\leq k).
\end{gather*}

Consider $\residual{\hat{e}_{k+1}}\theta$.  Since $\size(\,(\uinj_1 {\uep})\,)=0$, it follows along the lines of
subcase 1 of the $e=((\lam{z_1}e_0)\,e_1)$
case that $\residual{\hat{e}_{k+1}}\theta \leq 
\residual{e_{0}}\theta[z\mapsto (\uinj_1 {\uep})] 
\leq q_0\theta[z\mapsto (\uinj_1 {\uep})] 
\leq q_0[\residual{z}\gets 0]\theta$.

Next  consider $\residual{\hat{e}_{k}}\theta$.  
Let $\hat{v}_k = (\uinj_2 \upair{a_{k}}{u_{k+1}})$ 
and $\theta_0 =\theta[z\mapsto \hat{v}_k]$. 
We note that 
$\sharp_{\gamma_*}(\hat{v}_k) =  (\uinj_2 \upair{a_{k}}{\emptydag})$
and 
$\flat_{\gamma_*}(\hat{v}_k) =  (\uinj_2 \upair{\emptydag}{u_{k+1}})$,
and hence, $\nn{\gamma_*}(\hat{v}_k) \subseteq u_{k+1}$.
Thus, by Lemma~\ref{l:res}, 
$\residual{z}\theta_0  = 
   \size(\sharp_{\gamma_*}(\hat{v}_k) )$ which is =
$ \size(a_k) \leq q_1\theta$.
By Definition~\ref{d:resSize}\eqref{i:resSize},
$\residual{\hat{e}_{k}}\theta = \size(\sharp_{\safe{\gamma_0}}(u_k) 
\dotcup \nn{\safe{\gamma_0}}(u_k)\setminus SS(\fv(\hat{e}_k))\theta )$
which, by Lemma~\ref{l:ugh}, is $\leq$
\begin{align*}
& \size\left(\sharp_{\safe{\gamma_0}}(u_k) \dotcup (\nn{\safe{\gamma_0}}(u_k)\setminus SS(\fv(e_0))
\theta_0)\right)
\\ &\quad +\quad \size\left(\sharp_{\gamma_*}(\hat{v}_k) \dotcup (\nn{\gamma_*}(\hat{v}_k)\setminus SS(\fv(e_{k+1}))\theta)\right). 
\end{align*}
As in the $e=((\lam{z_1}e_0)\,e_1)$ case/subcase~3:
\begin{align*}
  &{\size\left(\sharp_{\safe{\gamma_0}}(u_k) \dotcup (\nn{\safe{\gamma_0}}(u_k)\setminus SS(\fv(e_0))
\theta_0)\right)}
  \\
  &\Quad1 = \residual{e_0}\theta_0
  \;\leq\; q_0\theta_0
  \;\leq\; q_0[\residual{z}\gets q_1]\theta .
\\[1ex]
  & \lefteqn{\size\left(\sharp_{\gamma_*}(\hat{v}_k) \dotcup (\nn{\gamma_*}(\hat{v}_k)\setminus SS(\fv(e_{k+1}))\theta)\right)}
  \\
  & \quad \leq \size(\sharp_{\gamma_*}(\hat{v}_k)) + 
   \size( (u_{k+1}\setminus SS(\fv(e_{k+1}))\theta)) & 
   \hbox{(since $\nn{\gamma_*}(\hat{v}_k)\subseteq u_{k+1}$).}
\end{align*}
We know that $\size(\sharp_{\gamma_*}(\hat{v}_k))\leq q_1\theta$.
Also we know that
$\sharp_{\safe{\gamma_0}}(u_{k+1})=\emptydag$
and $\nn{\safe{\gamma_0}}(u_{k+1}) = u_{k+1}$. So, 
$\size( (u_{k+1}\setminus SS(\fv(e_{k+1}))\theta))=\residual{e_{k+1}}\theta$
which by the previous induction step 
is $\leq q_0[\residual{z}\gets 0]\theta$. 
Thus, putting things together we have 
$\residual{e_k}\theta \leq (q_0[\residual{z}\gets q_1] + q_1 + 
q_0[\residual{z}\gets 0])\theta$.

It follows from a straightforward, if cluttered, downward induction
that for $i=k+1,\dots,1$:
\begin{align*} 
   \residual{e_i}\theta 
   \leq 
   \left((k+1-i)\cdot \left( \strut q_0[\residual{z}\gets q_1] + q_1\right) 
         + q_0[\residual{z}\gets 0] \right)\theta. 
\end{align*}
By the \emph{Fold}$_{\List_{\gamma_0}}$ evaluation rule,
$e\theta\yields u_1$.  Thus, it clearly suffices to take 
$q= q_1\cdot q_0[\residual{z}\gets q_1]+ (q_1)^2$.

\CASE $e=(\fold_{\delta} (\lam{z}e_0) \, e_1)$.
Thanks to the DP evaluation strategy, the 
$\delta=\List_{\gamma_1}$ case is pretty much the general case.
\end{proofsketch}

\begin{definition} Let $Y=\set{y_1,y_2,\dots }$ be a set of 
indeterminates.
  \begin{asparaenum}[(a)]
    \item A collection of terms $\SB$ is \emph{arithmetically closed}\sidenote{$Y$ is an implicit parameter.}
    iff $Y \cup \nat \subseteq \SB$ and, when $q_1,q_2\in\SB$,
    we have
    $q_1+q_2$,  $q_1\cdot q_2$, and $q_1[y_i\gets q_2] \in \SB$.
    \item Let $F$ be a collection of $\nat\to\nat$ functions.
    Let $\SB(F)$ = the smallest arithmetically closed set of terms
    that is also closed under applications of elements of $F$. 
\end{asparaenum}   
\end{definition}

Clearly, $\SB(\emptyset)$ = the polynomial terms over the indeterminates.

\pagebreak[4]
\begin{scholium}\ 
\begin{asparaitem}
  \item 
  	As inspection of the proof of Theorem~\ref{t:sizebnd} shows that
  	we can replace the set of polynomials over the indeterminates as 
  	size bounds with any arithmetically closed set of terms over the 
	indeterminates.  Consequently, if we take $\RSmi$ and add to it
	$I$, a set of initial functions, each of which has a bound in 
	$F$, some set of $\Nat\to\Nat$ functions, then a trivial extension
	of the proof of Theorem~\ref{t:sizebnd} shows that $\RSmi+I$
	is $\SB(F)$ size bounded. 
  \item 
  	Since $\RSmii= \RSmi + \set{\,\CS\muP \suchthat
	 \muP\not\equiv\mathbf{0}\,}$ and since 
  	$\size(\cs_\muP(v) ) \leq \size(v)= \id(\size(v))$ for all 
	type-$\muP$ values $v$, we have that $\RSmii$ terms are size bounded by
  	$\SB(\set{\id})$, i.e., polynomials.
  \item 
  	Let $K = \RSmi + \set{\exp \of\, \Nat\to\Nat}$ where 
	$\exp{\numeral{m}} = \numeral{2^m}$.   
	Since $\numeral{2^m} \leq 1+2^m$, we have that $K$ terms are size 
	bounded by $\SB(\set{\lam{m}(1+2^m)})$.
    $\SB(\set{\lam{m}(1+2^m)})$ turns out to define 
    a cofinal subset of
    $\SE_3$, the Kalm\'{a}r elementary recursive functions \citep{Rose84}.   
    With a bit more work one can show that $K$ is 
    $\SB(\set{\lam{m}(1+2^m)})$ cost bounded and that an analogue of
    the completeness result holds for $K$. 
  \item  
    By the same scheme used to define $K$, we can define
    sound and complete formalisms for each finite level of
    the Grzegorczyk hierarchy above $\SE_3$ \citep{Rose84} 
    and for the quasi-polynomial time\sidenote{I.e., 
    $\bigcup_{k>0}O(2^{(\log n)^{k}})$ time bounded.} functions.  
  \item 
    A key thing to note in all of these results is that the
    ramified types do their assigned job of ruling out diagonalizing
    definitions that climb out of the $\SB(F)$ size bounded functions
    (for the appropriate choice of $F$).    
\end{asparaitem}
\end{scholium}
\subsection{A polynomial cost bound}

\begin{theorem}
  Suppose  $\Gamma\entails e\of\gamma$ is an $\RSmi$-judgment
  with $\fv(e)=\set{x_1,\dots,x_k}$.
  Then there is a polynomial $p$ over indeterminates 
  $\residual{x_1},\dots,\residual{x_k}$ such that, for all 
  $\theta\of\Gamma$, we have
     $\mathit{cost}_{DP}(\Gamma\entails e\of\gamma\theta)
     \;\leq\; p \theta$. 
\end{theorem}

\begin{proofsketch}
Suppose $\theta\of\Gamma$. 
We proceed by strong induction on the derivation of 
$\Gamma\entails e\of\gamma$.  
We leave all but two cases as exercises for the reader. 

\CASE $e = ((\lam{z_1}e_0)\,e_1)$.  
By the induction hypothesis, there are  polynomial bounds $p_0$ and $p_1$
for the evaluation costs of $e_0$ and $e_1$, respectively.  By 
Theorem~\ref{t:sizebnd}, there is also 
a polynomial bound $q_1$ for $\residual{e_1}$. 
Suppose  $e_1\theta\yields v_1$ and $\theta_0=\theta[z_1\mapsto v_1]$.  Then
\begin{align*}
  &\lefteqn{\cost_{DP}(e\theta)}
  \\
  &\quad \leq 1+\cost_{DP}(e_1\theta)+\cost_{DP}(e_0\theta_0) 
  && \hbox{(by the $\lambda$-\emph{App} rule)}
  \\
  &\quad \leq 1+p_1\theta + p_0\theta_0
  && \hbox{(by IH on $e_0$ and $e_1$)}
  \\
  &\quad \leq 1+p_1\theta + p_0[\residual{z_1}\gets q_1]\theta
  && \hbox{(since $\residual{z_1}\theta_0 \leq q_1\theta$).}
\end{align*}
Thus it clearly suffices to take $p=1+p_0[\residual{z_1}\gets q_1]+p_1$.

\SCASE $e=(\fold_{\List_{\gamma_1}} (\lam{z}e_0) \; e_1)$.
By the induction hypothesis, there are polynomial bounds $p_0$ and $p_1$ 
for the evaluation costs of $e_0$ and $e_1$, respectively. 
By Lemma~\ref{l:costs}\eqref{i:DP<TD:seq}, 
without loss of generality we can use
the TD $\fold$-evaluation rule for this particular $\fold$. 
Recall that ${\List_{\gamma_1}}=\muP$  where 
$P X = \Unit + \gamma_1\times X$. By the 
$\fold$-evaluation rule, the base of the 
derivation tree for $e\theta$ is:
\begin{align*}
     \irule{((\lam{z}e_0)\,(\,g(\Dn{{\List_{\gamma_1}}} e_1)\,)) 
      \theta\yields v }%
  {(\fold_{{\List_{\gamma_1}}} (\lam{z}e_0)\;e_1)\theta \yields v}
\end{align*}
where $P \,(\fold_{{\List_{\gamma_1}}} (\lam{z}e_0)) \leadsto g$
and $g$ is, \emph{mutatis mutandis}, as in Example~\ref{ex:g}. 
Let $k$, $a_1,\dots,a_k$ and $u_{1},\dots,u_{k+1}$ be as in
the $e=(\fold_{\List_{\gamma_1}} (\lam{z}e_0) \; e_1)$ case 
in the proof of Theorem~\ref{t:sizebnd}.  Also let $q_1$ be
a polynomial bound for $\residual{e_1}$ as provided for by 
Theorem~\ref{t:sizebnd}.  A straightforward argument 
shows that $\cost(e\theta)$ is bounded by
\begin{align*}
   \left( p_1+\left(\strut p_0[\residual{z}\gets 0] + c\right) + 
   \sum_{i=1}^k \left(\strut p_0[\residual{z}\gets q_1] + c\right)\right)\theta,
\end{align*}
where $c$ is a constant (depending on $\List_{\gamma_1}$)
that bounds the cost of breaking down and rearranging values as 
dictated by $g$.   
Thus it clearly suffices to take $p=p_1 + q_1 \cdot (p_0[\residual{z}\gets q_1]+c)$.
\end{proofsketch}

\subsection{Noninterference}\label{S:noninterf}

Below we show that if one evaluates
$(\Gamma\entails e\of\gamma)\theta$
where $\gamma$ is normal, then the value produced
has no dependence upon the values of $e$'s safe-type free variables nor 
upon the safe parts of $e$'s mixed-type free variables.
We call this property \emph{normal invariance} (Definition~\ref{d:norm:inv}\eqref{i:norm:inv});
it is a form of  \emph{noninterference} \citep{Goguen:1982}.
These forms of noninterference results are common for tiered 
formalisms (e.g., \citep[Theorems 3 and 5]{Marion2011}), but the $\toNorm$-\emph{I} 
rule makes the proof of Theorem~\ref{t:norminvar} below a bit 
tricky.

\emph{Conventions:} For $m>2$: \
$\underline{(}v_1\underline{\ustrut,}v_2 \underline{\ustrut,}\,\dots 
   \underline{\ustrut,}\, v_m\underline{)}
   = 
   \underline{(}v_1\underline{\ustrut,}\,
   \underline{(}v_2\underline{\ustrut,}\,\dots 
   \underline{\ustrut,}\, v_m \underline{)}\,\underline{)}$.
   For dags $G$ and $G'$, $G\cong G'$ means that $G$ and $G'$ are 
   isomorphic.

\begin{definition}
Suppose $\Gamma = x_1\of\tau_1,\dots,x_n\of\tau_n$
and $\theta\of\Gamma$.  
Then:
\begin{gather*} 
   e\theta\Downarrow \underline{(}v_1\underline{\ustrut,}\,\dots 
   \underline{\ustrut,}\, v_n \underline{\ustrut,}\, v\underline{)} 
   \iffdef
   (x_1,\dots,x_n,e)\theta \yields \underline{(}v_1\underline{\ustrut,}\,\dots 
   \underline{\ustrut,}\, v_n \underline{\ustrut,}\, v\underline{)}.
\end{gather*}   
We call 
$\underline{(}v_1\underline{\ustrut,}\,\dots 
   \underline{\ustrut,}\, v_n \underline{\ustrut,}\, v\underline{)}$ 
the \emph{assembly} for $e\theta$.
\end{definition}

The \emph{assembly} for $e\theta$ contains 
$e\theta$'s value and
the value of 
each $x\in\dom(\Gamma)$  
\emph{together with} the sharing in and amongst 
these graphs.

\begin{definition} \label{d:norm:inv}
  Suppose $\Gamma=x_1\of\gamma_1,\dots,x_n\of\gamma_n$
  and $\theta,\theta'\of\Gamma$. 
  \begin{asparaenum}[(a)]
    \item \label{i:norm:iso:env}
      Suppose 
      $(x_1,\dots,x_n)\theta \yields v$ and 
      $(x_1,\dots,x_n)\theta' \yields v'$. 
      Then
     \begin{gather*}
       \theta\cong^\sharp\theta' \iffdef \;
         \sharp_{\gamma_1\times\dots\times\gamma_n}(v) 
         \cong
	     \sharp_{\gamma_1\times\dots\times\gamma_n}(v').
     \end{gather*}
  \item \label{i:norm:inv}
    An $\RSmi$ judgment 
   	$\Gamma\entails e \of\gamma$ is \emph{normal invariant}
    when, for all $\theta,\theta'\of\Gamma$ with
    $e\theta \Downarrow a$ and $e\theta' \Downarrow a'$, we have:
    \begin{gather} \label{e:norm:inv:impl}
	   \theta\cong^\sharp\theta'
	   \implies 
	   \sharp_{\gamma_1\times\dots\times\gamma_n\times\gamma}(a)\cong
	   \sharp_{\gamma_1\times\dots\times\gamma_n\times\gamma}(a').	   
	\end{gather}
\end{asparaenum}  
\end{definition}

If we view an environment as providing a
value context for a term, then \eqref{e:norm:inv:impl}
says that whenever the normal part of the context of $\theta$ 
and $\theta'$ agree, then the normal part of the assemblies $a$ 
and $a'$ agree.  Thus in particular, the normal part of $e$'s value
has no dependence upon the values of $e$'s safe-type free variables nor 
upon the safe parts of $e$'s mixed-type free variables.

\begin{theorem}  \label{t:norminvar}
  Each ground-type $\RSmi$ type judgment is normal-invariant. 
\end{theorem}

\begin{proofsketch}
The argument is a strong induction over the type 
derivation of 
$\Gamma\entails e\of\gamma$.
Suppose $\set{x_1,\dots,x_n} = \dom(\Gamma)$ and 
$\vec x$ abbreviates $x_1,\dots,x_n$. 
Also suppose $\theta,\,\theta'\of\Gamma$.  
\emph{Notation:}  Suppose 
$(\vec x)\theta\yields v$, $(\vec x)\theta'\yields v'$,
$e\theta\Downarrow a$, and $e\theta\Downarrow a'$, then:
\begin{asparaitem}
\item
    $(\sharp \vec{x})\theta \cong (\sharp \vec{x})\theta'$
    means $\sharp_{\gamma_1\times \dots \times \gamma_n}(v)\cong 
    \sharp_{\gamma_1\times \dots \times \gamma_n}(v')$, i.e., $\theta\cong^\sharp \theta'$.
\item
    $(\sharp \vec{x},\sharp e)\theta \cong (\sharp \vec{x},\sharp e)\theta'$
    means $\sharp_{\gamma_1\times \dots \times \gamma_n\times\gamma_0}(a)\cong 
    \sharp_{\gamma_1\times \dots \times \gamma_n\times\gamma_0}(a')$. 
\item
    $(\sharp \vec{x}, e)\theta \cong (\sharp \vec{x}, e)\theta'$
    means $\sharp_{\gamma_1\times \dots \times \gamma_n\times\norm{\gamma_0}}(a)\cong 
    \sharp_{\gamma_1\times \dots \times \gamma_n\times\norm{\gamma_0}}(a')$. 
\end{asparaitem}
Thus our goal is to show that $\theta\cong^\sharp \theta' \implies
(\sharp \vec{x},\sharp e)\theta \cong 
(\sharp \vec{x},\sharp e)\theta'$.

\CASE $\gamma$ is safe.  
This follows trivially.   Note that this case covers the cases of
$e=\safe{()}$,
$e=(\toSafe\,e_0)$, 
$e=(\Cs\delta e_0)$, 
$e=(\Ds\delta e_0)$, and 
$e=(\fold_\delta (\lam{z}e_0)\,e_1)$.

\medskip

The next case requires a  utility lemma.

\begin{lemma}\label{l:patch}  
  Suppose 
  $\Gamma\entails e\of\gamma$  where $\dom(\Gamma)=\set{x_1,\dots,x_m,y_1,\dots,y_n}$ and $\fv(e)\subseteq\set{y_{1},\dots,y_n}$.
  Then, for each $\theta,\theta'\of\Gamma$
  with $(\sharp\vec{x}, \vec{y})\theta 
  \cong (\sharp\vec{x}, \vec{y})\theta'$, we have that 
  $(\sharp \vec{x},\vec{y},e)\theta \cong (\sharp \vec{x},\vec{y},e)\theta'$.\sidenote{$(\sharp\vec{x}, \vec{y})\theta$ and 
  $(\sharp \vec{x},\vec{y},e)\theta$ are equivalent to 
  $(\sharp\vec{x}, (\vec{y}))\theta$ and 
  $(\sharp \vec{x},(\vec{y},e))\theta$, respectively.}
\end{lemma}

\begin{proof}
  Suppose $\theta,\theta'\of\Gamma$ with
  $(\sharp\vec{x}, \vec{y})\theta 
  \cong (\sharp\vec{x}, \vec{y})\theta'$.
  Let 
  $\hat{\theta}$ and $\hat{\theta}'$ be, respectively, 
  the restrictions of 
  $\theta$ and $\theta'$ to domain $\set{y_1,\dots,y_n}$. 
  Then clearly, $(\vec{y})\hat{\theta}\cong(\vec{y})\hat{\theta}'$.  Let $T$ be the
  evaluation derivation tree for $e\hat{\theta}$.
  It follows that the derivation trees for both $e\theta$ and $e\theta'$
  are isomorphic to $T$ in the sense that the structure of the trees
  and the expressions and values in the tree-node labels are the same,
  the only differences being in what the environments assign to the 
  $x_i$'s.   
  The lemma then follows by a straightforward induction argument on
  the nodes of $T$ that each expression at each of these nodes
  satisfies the lemma.  
\end{proof}

\CASE  $e=(\toNorm e_0)$. 
Since $\gamma$ is normal, 
$(\sharp\vec{x},\sharp(\toNorm e_0))\theta\cong 
 (\sharp\vec{x},(\toNorm e_0))\theta$ which, by
 the \emph{toNorm} evaluation rule, is $\cong   (\sharp\vec{x},e_0)\theta$.
Similarly, $(\sharp\vec{x},\sharp(\toNorm e_0))\theta'\cong (\sharp\vec{x},e_0)\theta'$.  
Suppose $\theta \cong^\sharp \theta'$. 
It is straightforward from Lemma~\ref{l:patch} and
the side-condition on \emph{I}-$\toNorm$
that $(\sharp\vec{x},e_0)\theta \cong (\sharp\vec{x},e_0)\theta'$. 
Therefore, $(\sharp \vec{x},\sharp e)\theta \cong 
(\sharp \vec{x},\sharp e)\theta'$ as required.  

\CASE $e=x$, a variable.  
Suppose $\theta \cong^\sharp \theta'$.
Since $x$ occurs in $\vec{x}$, it is immediate that
$(\sharp\vec{x},\sharp x)\theta \cong (\sharp\vec{x},\sharp x)\theta'$.

\CASE $e=((\lam{z_1}e_0)\,e_1)$ where 
$\Gamma,z_1\of\gamma_1\entails e_0\of\gamma$ and
$\Gamma\entails e_1\of\gamma_1$.
Let $\Gamma_0 = \Gamma,z_1\of\gamma_1$.  Also let 
${\theta}_0 = \theta[z_1\mapsto v_1]$ and 
${\theta}_0' = \theta'[z_1\mapsto v_1']$
where 
$e_1\theta\yields v_1$ and 
$e_1\theta'\yields v_1'$.
Suppose $\theta\cong^\sharp \theta'$.
Then by the induction hypothesis, 
$(\sharp\vec{x},\sharp e_1)\theta
\cong 
(\sharp\vec{x},\sharp e_1)\theta'$. 
Thus, ${\theta}_0\cong^\sharp{\theta}_0'$.
Then by the induction hypothesis, 
$(\sharp(\vec{x},z_1),\sharp e_0)\theta_0
\cong 
 (\sharp(\vec{x},z_1), \sharp e_0){\theta}_0'$. 
By the $\lambda$-\emph{App} evaluation rule, 
$e\theta \yields v\iff e_0{\theta}_0\yields v$ and 
$e\theta' \yields v'\iff e_0{\theta}_0'\yields v'$.
Hence by some assembly surgery, 
$(\sharp\vec{x},\sharp e)\theta
\cong 
 (\sharp\vec{x},\sharp e){\theta}'$ as required.

\CASE $e=(\Cn\delta e_0)$ where 
$\delta=\muP$ and $P\delta$ are normal, and 
$\Gamma\entails e_0\of P\delta$. 
Suppose 
$\theta\cong^\sharp\theta'$.
We need to show 
$(\sharp\vecx ,\sharp e)\theta \cong (\sharp\vecx ,\sharp e)\theta'$
which, since $\delta$ is normal, is equivalent to
$(\sharp\vecx , e)\theta \cong (\sharp\vecx , e)\theta'$.
By the IH on $e_0$ and $P\delta$'s normality,
$(\sharp \vecx ,e_0)\theta 
 \cong 
 (\sharp \vecx ,e_0)\theta'$. 
Suppose $e_0\theta\yields v_0$ and $e_0\theta'\yields v_0'$.
By the operational semantics, 
the result of 
evaluating $(\Cn\delta e_0)\theta$
is a fresh $\Cn\delta$-vertex with an out-edge to the root of $v_0$
and similarly with $(\Cn\delta e_0)\theta'$ and $v_0'$.
Hence since $(\sharp \vecx ,e_0)\theta \cong 
(\sharp \vecx ,e_0)\theta'$, it is evident that
$(\sharp \vecx , e)\theta 
 \cong 
 (\sharp \vecx , e)\theta'$.

\CASE  $e=(\Dn\delta e_0)$  where 
$\delta=\muP$ and $\gamma = P\delta$ are normal, and
$\Gamma\entails e_0\of\muP$. 
Suppose 
$\theta\cong^\sharp\theta'$.
We need to show 
$(\sharp\vecx ,\sharp e)\theta \cong (\sharp\vecx ,\sharp e)\theta'$
which, since $P\delta$ is normal, is equivalent to
$(\sharp\vecx , e)\theta \cong (\sharp\vecx , e)\theta'$.
By the IH on $e_0$ and $\delta$'s normality, 
$(\sharp \vecx ,e_0)\theta \cong (\sharp \vecx ,e_0)\theta'$. 
Suppose $e_0\theta\yields v_0$ and $e_0\theta'\yields v_0'$
By the operational semantics, $v_0$ (respectively, $v_0'$)
consists 
of a  $\Cn\delta$-vertex with an out-edge to the root of a $(P\delta)$-value
$\hat{v}_0$ (respectively, $\hat{v}_0'$). Moreover, 
it follows from the operational semantics
that $e\theta \yields \hat{v}_0$ (respectively, 
$e\theta' \yields \hat{v}_0'$).  Hence since 
$(\sharp \vecx ,e_0)\theta \cong (\sharp \vecx ,e_0)\theta'$,
it is evident that
$(\sharp\vecx ,e)\theta \cong (\sharp\vecx ,e)\theta'$.

\CASE  $e=\inj_i(e')$.
This is roughly a repeat of the $e=(\Cn\delta e_0)$  case. 

\CASE $e=(\Case e_0 \Of\, (\inj_1 z_1) \Rightarrow e_1;\,
          (\inj_2 z_2) \Rightarrow e_2)$,
where $\Gamma\entails e_0\of \gamma_1+\gamma_2$, 
$\Gamma,z_1\of \gamma_1\entails e_1\of\gamma$, and
$\Gamma,z_2\of \gamma_2\entails e_2\of\gamma$.  
Suppose that
$e_0\theta \yields (\uinj_j v_0)$, 
$e_0\theta' \yields (\uinj_k v_0')$,
$\theta_j =\theta[z_j \mapsto v_0]$, and
$\theta_k' =\theta[z_k \mapsto v_0']$.
Recall that the side-condition on $+$-\emph{E} requires that
$\gamma_1+\gamma_2$ is normal or $\gamma$ is safe.
As the $\gamma$-safe case is trivial, we take 
$\gamma_1+\gamma_2$ to be  normal.

Suppose that  $\theta\cong^\sharp \theta'$. 
By the IH on $e_0$, 
$(\sharp \vecx ,\sharp e_0)\theta \cong (\sharp \vecx ,\sharp e_0)\theta'$,
which, since $\gamma_1+\gamma_2$ is normal, is equivalent to
$(\sharp \vecx , e_0)\theta \cong (\sharp \vecx , e_0)\theta'$.
Thus we have that 
$(\uinj_j v_0)   \cong (\uinj_k v_0')$, and so 
$j=k$ and the evaluations of $e\theta$ and $e\theta'$ both
take the $j$-branch of the $\Case$.
Also, since $(\sharp \vecx , e_0)\theta \cong (\sharp \vecx , e_0)\theta'$,
it follows along the lines of the argument for the $e=(\Dn\muP e_0)$ case
that $\theta_j\cong^\sharp \theta_j'$.  Hence, 
by the IH for $e_j$, 
$(\sharp (\vecx,z_j),\sharp e_i)\theta_j \cong 
 (\sharp (\vecx ,z_j),\sharp e_i)\theta_j'$.
But then, it follows along the lines of the argument 
for the $e=((\lam{z} e_0)\,e_1)$ case that 
$(\sharp \vecx ,\sharp e)\theta \cong 
 (\sharp \vecx ,\sharp e)\theta'$ as required.

\CASE $e=()$.  
This  follows trivially. 

\CASE $e=\proj_i(e_0)$ where $\Gamma\entails e_0\entails \gamma_1\times\gamma_2$.  Without loss of generality, suppose
that $\gamma_i$ is not safe. Thus, $\gamma_1\times\gamma_2$ is
not safe either. 
Suppose 
$\theta\cong^\sharp\theta'$. 
We need to show $(\sharp\vecx ,\sharp e)\theta \cong
 (\sharp\vecx ,\sharp e)\theta'$. 
Suppose 
$e\theta\yields v$, $e\theta'\yields v'$, 
$e_0\theta\yields \upair{v_1}{v_2}$ and $e_0\theta'\yields \upair{v_1'}{v_2'}$.  
Since $\gamma_1\times\gamma_2$ is not safe, 
by Definition~\ref{d:spans}\eqref{i:sp:norm}, neither 
$\sharp_{\gamma_1\times\gamma_2}(\upair{v_1}{v_2})$ nor
$\sharp_{\gamma_1\times\gamma_2}(\upair{v_1'}{v_2'})$ is $\emptydag$, 
and hence,  
$\sharp_{\gamma_i}(v) = \sharp_{\gamma_i}(v_i)$ and 
$\sharp_{\gamma_i}(v') = \sharp_{\gamma_i}(v_i')$.
By the IH for $e_0$, 
$(\sharp\vecx ,\sharp e_0)\theta \cong
 (\sharp\vecx ,\sharp e_0)\theta'$. 
It thus follows that 
$(\sharp\vecx ,\sharp e)\theta \cong
 (\sharp\vecx ,\sharp e)\theta'$ as required.

\SCASE $e=(x_j,x_k)$, where $\gamma=\gamma_i\times\gamma_j$
for $j,k\in\set{1,\dots,n}$.  
The case follows trivially when $\gamma_i\times\gamma_j$ is safe, 
so suppose $\gamma_i\times\gamma_j$ is not safe.
Suppose 
$\theta\cong^\sharp\theta'$. 
We need to show $(\sharp\vecx ,\sharp e)\theta \cong
 (\sharp\vecx ,\sharp e)\theta'$.
Suppose $e\theta\yields v$ and $e\theta'\yields v'$.
Since $\gamma_i\times\gamma_j$ is not safe,
neither $\sharp_{\gamma} v$ nor $\sharp_{\gamma} v'$ 
is $\emptydag$. Hence by a slight abuse of notation,
$(\sharp\vec{x},\sharp e)\theta \cong
 (\sharp\vec{x},(\sharp x_j,\sharp x_k))\theta$ 
and 
$(\sharp\vec{x},\sharp e)\theta' \cong
 (\sharp\vec{x},(\sharp x_j,\sharp x_k))\theta'$. 
Since $\theta\cong^\sharp\theta'$, it follows
that 
$(\sharp\vec{x},(\sharp x_j,\sharp x_k))\theta \cong
 (\sharp\vec{x},(\sharp x_j,\sharp x_k))\theta'$, and thus,
$(\sharp\vec{x},\sharp e)\theta \cong (\sharp\vec{x},\sharp e)\theta'$
as required. 

\CASE $e=(\hat{e}_1,\hat{e}_2)$, where $\gamma=\gamma_1\times\gamma_2$.  
Let 
\begin{gather*}
e_* = ((\lam{z_1}((\lam{z_2}(z_1,z_2)) \; \hat{e}_1)\; \hat{e}_2)).
\end{gather*}
For all $\theta''\of\Gamma$, 
$(\vec{x},e)\theta''\cong (\vec{x},e_*)\theta''$.
Then a combination of the arguments for the prior special case and the 
$e=((\lam{z}e_0)\, e_1)$
case shows this case. 
\end{proofsketch}

\end{document}